\documentclass[1p]{finisn-arxiv}
\usepackage{palatino}
\usepackage[all]{xy}

\usepackage{float}
\floatstyle{boxed} \restylefloat{figure}

\usepackage{amsfonts}
\usepackage{amsmath}
\usepackage{amsthm}
\usepackage[english]{babel}
\usepackage[colorlinks,urlcolor=black]{hyperref}
\usepackage{prooftree}
\usepackage{stmaryrd}
\usepackage{amssymb} 
\usepackage{graphicx}
\usepackage{mdwlist}

\usepackage{array}
\newcolumntype{L}[1]{>{$}p{#1}<{$}}
\newcolumntype{C}[1]{>{\centering$}p{#1}<{$}}
\newcolumntype{R}[1]{>{\raggedleft$}p{#1}<{$}}
\newcommand\maketab[2]{\newenvironment{#1}{\begin{quote}\noindent\begin{tabular}{#2}}{\end{tabular}\end{quote}}\newenvironment{noquote#1}{\noindent\begin{tabular}{#2}}{\end{tabular}}}

\newtheoremstyle{jamiestyle}
  {4pt}
  {0pt}
  {\it}
  {0pt}
  {\bf}
  {.}
  { }
  {}
\theoremstyle{jamiestyle}
\newtheorem{thrm}{Theorem}[section]
\newtheorem{prop}[thrm]{Proposition}
\newtheorem{lemm}[thrm]{Lemma}
\newtheorem{corr}[thrm]{Corollary}
\newtheoremstyle{jamienfstyle}
  {4pt}
  {0pt}
  {\normalfont}
  {0pt}
  {\bf}
  {.}
  { }
  {}
\theoremstyle{jamienfstyle}
\newtheorem{nttn}[thrm]{Notation}
\newtheorem{defn}[thrm]{Definition}
\newtheorem{xmpl}[thrm]{Example}
\newtheorem{rmrk}[thrm]{Remark}

\usepackage{fancybox}
\newlength{\mylength}
\newenvironment{frameqn}%
{\setlength{\fboxsep}{6pt}
\setlength{\mylength}{\linewidth}
\addtolength{\mylength}{-2\fboxsep}%
\addtolength{\mylength}{-2\fboxrule}%
\Sbox
\minipage{\mylength}%
\setlength{\abovedisplayskip}{0pt}%
\setlength{\belowdisplayskip}{0pt}%
$$}%
{$$\endminipage\endSbox
\[\fbox{\TheSbox}\]}

\newenvironment{frametxt}%
{\setlength{\fboxsep}{5pt}
\setlength{\mylength}{\textwidth}%
\addtolength{\mylength}{-2\fboxsep}%
\addtolength{\mylength}{-2\fboxrule}%
\Sbox
\minipage{\mylength}%
\setlength{\abovedisplayskip}{0pt}%
\setlength{\belowdisplayskip}{0pt}%
}%
{\endminipage\endSbox
\[\fbox{\TheSbox}\]}

\newcommand\fini{\mathrm{fin}}
\newcommand{\mn}{\text{-}}
\newcommand{\aeq}{=_{\scriptstyle\alpha}}

\newcommand{\deffont}[1]{\textbf{#1}}

\newcommand{\f}[1]{\ensuremath{\text{$\mathit{#1}$}}}
\newcommand{\fcomp}{{\circ}}

\newcommand{\rulefont}[1]{\ensuremath{(\mathbf{#1})}}

\newcommand{\ssm}{:=}
\newcommand{\supp}{\f{supp}}
\newcommand{\tf}[1]{\mathsf{#1}}
\newcommand{\theory}[1]{\ensuremath{\mathsf{#1}}}

\newcommand\shift{\delta}
\newcommand\lstel[1]{a_{#1}}
\newcommand\limp{{\Rightarrow}}
\newcommand\liff{\Leftrightarrow}

\newcommand{\fa}{\f{fa}}
\newcommand{\fv}{\f{fv}}

\newcommand\fix{\f{fix}}
\usepackage{marvosym}\newcommand\at{\MVAt}
\newcommand{\atoms}{\f{atoms}}
\newcommand\nontriv{\f{nontriv}}
\newcommand\pmss[1]{\atomsdown}
\newcommand\pmssC{\f{pmss}}

\newcommand\constant{C}
\newcommand\sort{\f{sort}}
\newcommand\den[1]{{\hspace{-.1ex}\scalebox{.65}{$#1$}}}
\newcommand\hden{\den{\interp H}} 
\newcommand\flden{\den{[l]\interp H}} 
\newcommand\fmden{\den{[m]\interp H}} 
\newcommand\iden{\den{\interp I}} 

\newcommand{\idenot}[2]{\denot{\interp I}{#1}{#2}}
\usepackage{mathrsfs} 
\newcommand\interp[1]{\ensuremath{\mathscr #1}}

\newcommand\ns{\mathsf}
\newcommand{\denot}[3]{\llbracket #3 \rrbracket_{\scalebox{.6}{$#2$}}^{\hspace{-.1ex}\scalebox{.55}{$#1$}}}
\newcommand{\hdenot}[1]{\denot{\interp H}{}{#1}}
\newcommand{\fdenot}[1]{\denot{\interp F}{}{#1}}
\newcommand{\fmdenot}[1]{\denot{[m]\interp H}{}{#1}}
\newcommand{\act}{{\cdot}}

\newcommand\ment{\vDash}

\newcommand\mentfin{\vDash_{\scalebox{.6}{\it fin}}}
\newcommand\mentint{\vDash_{\dd}}

\newcommand\mone{\ensuremath{{\text{-}1}}}

\newcommand{\id}{{id}}

\newcommand\Ldd{\mathbb L^{\hspace{-.3ex}\dd}}
\newcommand\Ld{\mathbb L^{\hspace{-.3ex}{}_{{}^<}}}
\newcommand\atomsdown{\mathbb A^{\hspace{-.3ex}{}_{{}^<}}}
\newcommand\dd{{\scalebox{.6}{$<\hspace{-3pt}<$}}}
\newcommand\atomsdd{\mathbb A^{\hspace{-.3ex}\dd}}
\newcommand\atomsup{\mathbb A^{\hspace{-.3ex}{}_{{}^>}}} 
\newcommand\basesort{\tau}

\newcommand\Forall[1]{\forall #1.}
\newcommand\Exists[1]{\exists #1.}

\usepackage{graphics} 
\usepackage{type1cm}
\usepackage{eso-pic}

\makeatletter
\AddToShipoutPicture{%
            \setlength{\@tempdimb}{.5\paperwidth}%
            \setlength{\@tempdimc}{.5\paperheight}%
            \setlength{\unitlength}{1pt}%
            \put(\strip@pt\@tempdimb,\strip@pt\@tempdimc){%
        \makebox(0,0){\rotatebox{55}{\textcolor[gray]{0.98}%
        {\fontsize{5cm}{5cm}\selectfont{}}}}%
            }%
}
\makeatother

\author{Murdoch J. Gabbay}
\address{\href{http://www.gabbay.org.uk}{\it http://www.gabbay.org.uk}}
\title{Finite and infinite support in nominal algebra and logic: \\ nominal completeness theorems for free}
\date{}

\begin{document}

\begin{abstract}
By operations on models we show how to relate completeness with respect to permissive-nominal models to completeness with respect to nominal models with finite support.
Models with finite support are a special case of permissive-nominal models, so the construction hinges on generating from an instance of the latter, some instance of the former in which sufficiently many inequalities are preserved between elements.
We do this using an infinite generalisation of nominal atoms-abstraction.

The results are of interest in their own right, but also, we factor the mathematics so as to maximise the chances that it could be used off-the-shelf for other nominal reasoning systems too.
Models with infinite support can be easier to work with, so it is useful to have a semi-automatic theorem to transfer results from classes of infinitely-supported nominal models to the more restricted class of models with finite support.

In conclusion, we consider different permissive-nominal syntaxes and nominal models and discuss how they relate to the results proved here.
\end{abstract}

\begin{keyword}
Permissive-nominal techniques, infinite support, nominal algebra, permissive-nominal logic, completeness.
\end{keyword}

\maketitle

\newpage
\tableofcontents
\section{Introduction}
\label{sect.introduction}

Nominal techniques are an approach to variables in syntax and semantics which give variables denotational reality as \emph{names}.
The semantics underlying nominal techniques are \emph{nominal sets} \cite{gabbay:newaas-jv}, which identify variable symbols with names or (for set theorists) \emph{urelemente}. 
We may call names/urelemente \emph{atoms} and we write the set of all atoms as $\mathbb A$.

According to nominal techniques, syntax and semantics both `contain' atoms, in a sense made formal by a notion of \emph{support} (see Definition~\ref{defn.nominal.set}).

The original applications of nominal sets and nominal terms \cite{gabbay:newaas-jv,gabbay:nomu-jv} admitted only finite support (the interested reader can find more applications listed on \cite{mulligan:onlnb}). 

Permissive-nominal terms and models generalise this by allowing infinite support (based on a set of finitely representable but still infinite supporting sets called \emph{permission sets}).
Precise definitions will come later. 
For the benefit of the reader already familiar with nominal techniques we give a simple schematic for how this fits together:
$$
\begin{array}{c@{\quad}c@{\quad}c}
\text{nominal sets} &\leftrightarrow& \text{nominal terms}
\\
\subseteq & & \subseteq
\\
\text{permissive-nominal sets} &\leftrightarrow& \text{permissive-nominal terms}
\end{array}
$$
Both models and syntax seem better-behaved in the permissive case:
we avoid the conditional reasoning typical of more traditional finitely-supported nominal techniques.\footnote{For instance, `nominal algebra' uses equations conditional on freshness constraints saying that `$a$ is fresh for $X$' \cite{gabbay:nomuae}, whereas `permissive-nominal algebra' uses just equations \cite{gabbay:nomtnl}.} 
This makes it possible to unify the semantic and syntactic notions of $\alpha$-equivalence and freshness, to `just quotient' terms by $\alpha$-equivalence, and to cleanly add universal quantification. 
Some complex mathematical proofs become dramatically simpler.
Precise examples are cited in the Conclusions of this paper.

So permissive-nominal techniques are arguably nicer to work with, but `ordinary' nominal techniques are arguably more elementary (no infinities to confuse the reader)---and they are sufficient for many applications.

We indicate subset inclusions in the schematic above because models with finite support are special cases of models with infinite support, and it has been shown by arguments on syntax how to map from `ordinary' nominal syntax to permissive-nominal syntax \cite[Section~4]{gabbay:perntu-jv}. 

But what about the other way around?

In this paper, we explore models with differently-sized sets of atoms, give constructions to move from `larger' to `smaller' support, and test when these size transformations can and cannot be internally detected by the logics concerned.
The main two results are Theorems~\ref{thrm.fin.nonfin} and~\ref{thrm.pnl.fin.nonfin}---these follow from two technical results, Theorem~\ref{thrm.pull.out.l} and Lemma~\ref{lemm.l.varsigma}.

Because our arguments are based on models, it is fairly easy to apply them to different syntaxes.
In this paper we use the two examples studied in previous work by the author and others: nominal algebra \cite{gabbay:nomuae} (an equality reasoning system whose term language is nominal terms) and permissive-nominal logic (ditto, for first-order logic) \cite{gabbay:pernl-jv}.
See also a recent survey paper, which covers both of these in a uniform presentation \cite{gabbay:nomtnl}.

\subsection*{Structure of the paper}
\begin{itemize*}
\item
In Section~\ref{sect.perns} we briefly introduce permissive-nominal sets, with examples.
These will be our semantic universe in this paper; nominal sets from \cite{gabbay:newaas-jv} are a special case.
\item
In Section~\ref{sect.nominal.terms.syntax} we introduce permissive-nominal terms: signatures, terms, $\alpha$-equivalence, and their interpretation in permissive-nominal sets. 
The critical definition is Definition~\ref{defn.interpret.terms}, which maps from syntax to semantics.
\item
Section~\ref{sect.models.with.finite.support} shows how to reduce the size of the support of a interpretation with `large' support, to obtain a interpretation with `smaller' support.
This requires some interesting technical constructions.
Notably, we consider atoms-abstraction by a list of atoms $[l]x$ (Definition~\ref{defn.abstraction.X}), and a permutative notion of restricting a permutation $\pi/S$ (Definition~\ref{defn.pi.S}).
\item
In Section~\ref{sect.two.comm} are three technical commutation results: the common theme is that reducing the size of the support of a interpretation commutes with the structure of that interpretation.
\item
Section~\ref{sect.completeness.fin} proves our first main theorem, that permissive-nominal algebra is complete over finitely-supported interpretations (Theorem~\ref{thrm.fin.nonfin}).
\item
Section~\ref{sect.pnl} introduces a novel notion of `medium support' (Definition~\ref{defn.medium})
and proves our second main theorem, that permissive-nominal logic over interpretations with medium support has the same validity as over interpretations with finite support (Theorem~\ref{thrm.pnl.fin.nonfin}).
We discuss what this means in Subsection~\ref{subsect.pnl.relevance}.
\item
Section~\ref{sect.shift} discusses how the precise design of permission sets and permutations affects the proofs of this paper.
We find that the results are delicate: even quite small changes can break the proofs (Propositions~\ref{prop.upgrade.zero.fail} and~\ref{prop.upgrade.fail}).
\item
We conclude with a technical discussion of our results, related work, and future work.
\end{itemize*}

\section{Permissive-nominal sets}
\label{sect.perns}
We start with the basic definitions of permission sets, permissive-nominal sets, and then we give some examples.

\subsection{Atoms, permutations, permission sets}

\begin{defn}
\label{defn.NZ}
Write $\mathbb N=\{0,1,2,3,\ldots\}$ for the natural numbers.
and 
$\mathbb Z=\{0,\text{-}1,1,\text{-}2,2,\ldots\}$ for the integers.
\end{defn} 

\begin{defn}
\label{defn.atoms}
For each $i\in\mathbb N$ fix a pair of disjoint countably infinite sets of \deffont{atoms} $\atomsdown_i$ and $\atomsup_i$.
Write 
$$
\atomsdown=\bigcup\atomsdown_i,
\quad
\atomsup=\bigcup\atomsup_i,
\quad
\mathbb A_i=\atomsdown_i\cup\atomsup_i,
\quad\text{and}\quad
\mathbb A=\atomsdown\cup\atomsup.
$$
$a,b,c,\ldots$ will range over \emph{distinct} atoms: we call this the \deffont{permutative} convention.
\end{defn}

\begin{defn}
\label{defn.swap}
\label{def.permutation}
\label{def.nontriv}
Given $a,b\in\mathbb A_i$ for some $i\in\mathbb N$ write $(a\ b)$ for the \deffont{swapping} bijection on atoms mapping $a$ to $b$, $b$ to $a$, and any other $c\in\mathbb A\setminus\{a,b\}$ to $c$.

If $\pi$ is a bijection on atoms define $\nontriv(\pi)=\{a\mid \pi(a)\neq a\}$.
 
Write $\mathbb P_{\fini}$ for the group of bijections (finitely) generated by swappings, and call these bijections \deffont{permutations}.

Write $\pi \fcomp \pi'$ for the \deffont{composition} of $\pi$ and $\pi'$ (so $(\pi\fcomp\pi')(a)=\pi(\pi'(a))$).
Write $\id$ for the \deffont{identity} permutation (so $\id(a)=a$ always). 
\end{defn}

\begin{lemm}
A bijection $\pi$ on atoms is a permutation if and only if
\begin{itemize*}
\item
$a\in\mathbb A_i$ if and only if $\pi(a)\in\mathbb A_i$.
\item
$\nontriv(\pi)=\{a\mid \pi(a)\neq a\}$ is finite.
\end{itemize*}
\end{lemm}

\begin{defn}
\label{defn.pointwise}
If $A\subseteq\mathbb A$ define the \deffont{pointwise} action by $\pi\act A=\{\pi(a)\mid a\in A\}$.

\label{defn.the.comb}
A \deffont{permission set} $S$ is a set of the form $\pi\act \atomsdown$.
$S,T$ will range over permission sets.
\end{defn}

The choices made in Definitions~\ref{def.permutation} and~\ref{defn.the.comb} make Theorems~\ref{thrm.fin.nonfin} and~\ref{thrm.pnl.fin.nonfin} work.
These choices are possible within the framework of \cite{gabbay:nomtnl}.

\subsection{Permissive-nominal sets}

\begin{defn}
\label{defn.perm.set} 
A \deffont{set with a permutation action} $\ns X$ is a pair $(|\ns X|,\act)$ of 
a \deffont{carrier set} $|\ns X|$ and 
a group action on the carrier set $(\mathbb P_{\fini}\times |\ns X|)\to |\ns X|$, written infix as $\pi\act x$.\footnote{So, $\id\act x=x$ and $\pi\act(\pi'\act x)=(\pi\fcomp\pi')\act x$ for every $\pi$ and $\pi'$ and every $x\in|\ns X|$.}

\label{defn.support}
Say $A\subseteq\mathbb A$ \deffont{supports} $x\in |\ns X|$ when for every (finite) permutation $\pi\in\mathbb P_{\fini}$, if $\pi(a) =a$ for all $a \in A$ then $\pi\act x =x$.
\end{defn}

\begin{frametxt}
\begin{defn}
\label{defn.nominal.set} 
A \deffont{permissive-nominal set} is a set with a permutation action such that every element has a unique least supporting set $\supp(x)$ such that $\supp(x)\subseteq S$ for some permission set $S$.
We call this the \deffont{support} of $x$. 

$\ns X$, $\ns Y$ will range over permissive-nominal sets.
\end{defn}
\end{frametxt}

In fact, if $x\in|\ns X|$ has \emph{some} supporting set $A\subseteq S$, then it has a \emph{least} one; see e.g. \cite[Theorem~4.3]{gabbay:pernl}.

\begin{defn}
\label{defn.restrict}
If $\pi$ is a permutation and $A\subseteq\mathbb A$ write $\pi|_A$ for the \deffont{restriction} of $\pi$ to $A$.
This is the partial function such that $\pi|_A(a)=\pi(a)$ when $a\in A$, and is undefined otherwise.
\end{defn}

\begin{lemm}
\label{lemm.supp.restricted}
Suppose $\ns X$ is a nominal set.
Suppose $x\in|\ns X|$ and $A\subseteq\mathbb A$ supports $x$.

Then $\pi|_A=\pi'|_A$ implies $\pi\act x=\pi'\act x$.
\end{lemm}
\begin{proof}
From the definition of support, considering $\pi^\mone\fcomp\pi'$.
\end{proof}

\begin{lemm}
\label{lemm.supp.pi.x}
Suppose $\ns X$ is a permissive-nominal set and $x\in |\ns X|$.
Then $\supp(\pi\act x)=\pi\act\supp(x)$.
\end{lemm}
\begin{proof}
By a routine calculation using the group action. 
\end{proof}

We conclude with a useful condition for checking whether $a\in\supp(x)$:
\begin{corr}
\label{corr.notinsupp}
Suppose $\ns X$ is a permissive-nominal set and $x\in |\ns X|$.
Suppose $b\not\in\supp(x)$.
Then $(b\ a)\act x= x$ if and only if $a\not\in\supp(x)$.
\end{corr}
\begin{proof}
Suppose $b\not\in\supp(x)$.
The right-to-left implication is by the definition of support. 
For the left-to-right implication, we prove the contrapositive.
Suppose $a\in\supp(x)$.
By Lemma~\ref{lemm.supp.pi.x} $\supp((b\ a)\act x)=(b\ a)\act \supp(x)$.
By our suppositions, $(b\ a)\act\supp(x)\neq \supp(x)$.
It follows that $(b\ a)\act x\neq x$. 
\end{proof}

\subsection{Examples} 
\label{subsect.pn.sets.examples}
We briefly consider examples of permissive-nominal sets, which will be useful shortly.


\begin{defn}
\label{defn.atoms.perm}
$\mathbb A$ the set of atoms can be considered a permissive-nominal set with a natural permutation action $\pi\act a=\pi(a)$.


In the case of $\mathbb A$ 
only, we will be lax about the distinction between the set, and the permissive-nominal set with its natural permutation action.
\end{defn}


\begin{defn}
\label{defn.abstraction.sets}
Suppose $\ns X$ is a permissive-nominal set and $\mathbb A_\nu$ is a set of atoms.
Suppose $x\in|\ns X|$ and $a\in\mathbb A_\nu$.
Define \deffont{atoms-abstraction} $[a]x$ and $[\mathbb A_\nu]\ns X$ by:
\begin{frameqn}
\begin{array}{r@{\ }l}
[a]x =& \{(a,x)\}\cup \{(b,(b\ a)\act x)\mid b\in\mathbb A_\nu{\setminus} \f{supp}(x)\}
\\
|[\mathbb A_\nu]\ns X| =& \{[a]x\mid a\in\mathbb A_\nu,\ x\in|\ns X|\} 
\\
\pi\act [a]x =& [\pi(a)]\pi\act x
\end{array}
\end{frameqn}
\end{defn}
\noindent (Compare Definition~\ref{defn.abstraction.sets} with Definition~\ref{defn.abstraction.X}.)

\begin{rmrk}
In the definition of $[a]x$ in Definition~\ref{defn.abstraction.sets} recall that by our permutative convention $b\neq a$.
An equivalent and more compact way of writing this is $[a]x=\{(\pi(a),\pi\act x)\mid \pi\in\f{fix}(\f{supp}(x){\setminus}\{a\})\}$ where $\f{fix}(A)=\{\pi\mid\Forall{a{\in}A}\pi(a)=a\}$.
\end{rmrk}

\begin{lemm}
\label{lemm.supp.abstraction}
\begin{enumerate*}
\item
$[\mathbb A_\nu]\ns X$ is a permissive-nominal set.
\item
$[a]x{=}[a]x'$ if and only if $x{=}x'$, for $a{\in}\mathbb A_\nu$ and $x{\in} |\ns X|$.
\item
$[a]x{=}[a']x'$ if and only if $a'{\not\in}\f{supp}(x)$ and $(a'\, a)\act x{=}x'$, for $a,a'{\in}\mathbb A_\nu$ and $x,x'{\in}|\ns X|$.
\end{enumerate*}
\end{lemm}

\begin{defn}
\label{defn.times}
If $\ns X_i$ are permissive-nominal sets for $1\leq i\leq n$ then define $\ns X_1\times\ldots\times \ns X_n$ by:
$$
\begin{array}{r@{\ }l}
|\ns X_1\times\ldots\times\ns X_n|=&|\ns X_1|\times\ldots\times|\ns X_n|
\\
\pi\act (x_1,\ldots,x_n)=&(\pi\act x_1,\ldots,\pi\act x_n)
\end{array}
$$
\end{defn}

\begin{lemm}
\label{lemm.properties.of.support}
\begin{itemize*}
\item
$\f{supp}(a)=\{a\}$.
\item
$\f{supp}([a]x)=\f{supp}(x)\setminus\{a\}$.
\item
$\f{supp}((x_1,\ldots,x_n))=\bigcup\{\f{supp}(x_i)\mid 1\leq i\leq n\}$.
\end{itemize*}
\end{lemm}
\begin{proof}
Proofs are as in \cite{gabbay:newaas-jv} or \cite{gabbay:fountl}. 
\end{proof}

\section{Permissive-nominal terms syntax and its interpretation}
\label{sect.nominal.terms.syntax}

\subsection{Signatures}

\begin{defn}
\label{defn.sort.sig}
A \deffont{sort-signature} is a tuple $(\mathcal A,\mathcal B)$ of \deffont{name} and \deffont{base} sorts $\mathcal A\subseteq \mathbb N$ and $\mathcal B$.

$\nu$ will range over name sorts; $\basesort$ will range over base sorts.

A \deffont{sort language} is defined by
\begin{frameqn}
\alpha ::= \nu \mid \basesort \mid 
(\alpha,\ldots,\alpha) \mid [\nu]\alpha
.
\end{frameqn}
\end{defn}

\begin{frametxt}
\begin{defn}
\label{defn.term.signature}
A \deffont{term-signature} over a sort-signature $(\mathcal A,\mathcal B)$ is a tuple $(\mathcal C,\mathcal X,\mathcal F,\f{ar},\pmssC)$ where:
\begin{itemize*}
\item
$\mathcal C$ is a set of \deffont{constants}.
\item
$\mathcal X$ is a set of \deffont{unknowns}. 
\item
$\mathcal F$ is a set of \deffont{term-formers}. 
\item 
$\f{ar}$ assigns 
\begin{itemize*}
\item
to each constant $\constant\in\mathcal C$ a base sort $\tau$ which we may write $\sort(C)$, 
\item
to each unknown $X\in\mathcal X$ a sort $\alpha$ which we write may $\sort(X)$, and 
\item
to each ${\tf f\in\mathcal F}$ a \deffont{term-former arity} $(\alpha)\tau$, where
$\alpha$ and $\tau$ are in the sort-language determined by $(\mathcal A,\mathcal B)$.
\end{itemize*}
\item
$\pmssC$ assigns to each constant a set $\pmssC(C)\subseteq\atomsdown$. 
\end{itemize*}
\label{defn.signature}
A \deffont{(nominal terms) signature} $\Sigma$ is then a tuple $(\mathcal A,\mathcal B,\mathcal C,\mathcal X,\mathcal F,\f{ar},\pmssC)$.
\end{defn}
\end{frametxt}

We may write $((\alpha_1,\ldots,\alpha_n))\tau$ just as $(\alpha_1,\ldots,\alpha_n)\tau$.

\subsection{Terms}

\begin{defn}
\label{defn.syntax}
\label{defn.terms}
For each signature $\Sigma=(\mathcal A,\mathcal B,\mathcal C,\mathcal X,\mathcal F,\f{ar},\pmssC)$, define \deffont{terms} over $\Sigma$ by: 
\begin{frameqn}
\begin{array}{c@{\qquad}c@{\qquad}c}
\begin{prooftree}
(a\in\mathbb A_\nu,\ \nu\in\mathcal A)
\justifies
a:\nu
\end{prooftree}
&
\begin{prooftree}
(\sort(\constant)=\tau)
\justifies
\pi\act\constant:\tau
\end{prooftree}
&
\begin{prooftree}
(\sort(X)=\alpha)
\justifies
\pi\act X:\alpha
\end{prooftree}
\\[4ex]
\begin{prooftree}
r:\alpha\quad (\f{ar}(\tf f)=(\alpha)\tau)
\justifies
\tf f(r):\tau
\end{prooftree}
&
\begin{prooftree}
r_1:\alpha_1 \ \ldots\ r_n:\alpha_n
\justifies
(r_1,\ldots,r_n):(\alpha_1,\ldots,\alpha_n)
\end{prooftree}
&
\begin{prooftree}
r:\alpha\quad (a\in\mathbb A_\nu,\ \nu\in\mathcal A)
\justifies
[a]r:[\nu]\alpha
\end{prooftree}
\end{array}
\end{frameqn}
\end{defn}

We may write $\tf f((r_1,\ldots,r_n))$ as $\tf f(r_1,\ldots,r_n)$.

\maketab{tab2}{@{\hspace{0em}}R{8em}@{\ }L{6em}@{\ }R{8em}@{\ }L{10em}}
\begin{defn}
\label{defn.fa}
Define \deffont{free atoms} and the \deffont{permutation action}, and \deffont{free variables} on terms $r$ as follows:
\begin{tab2}
\fa(a) =& \{a\}
&
\fa(\tf f(r)) =& \fa(r) 
\\
\fa(\pi\act C) =& \pi\act\pmssC(C) 
&
\fa((r_1,\ldots,r_n)) =& \bigcup_{1\leq i\leq n} \fa(r_i) 
\\
\fa(\pi\act X) =& \pi\act\pmss{X} 
&
\fa([a]r) =& \fa(r){\setminus}\{a\} 
\\[2ex]
\pi\act a=&\pi(a)
&
\pi\act\tf f(r)=&\tf f(\pi\act r)
\\
\pi\act (\pi'\act C)=&(\pi\fcomp\pi')\act C
&
\pi\act(r_1,\ldots,r_n)=&(\pi\act r_1,\ldots,\pi\act r_n)
\\
\pi\act (\pi'\act X)=&(\pi\fcomp\pi')\act X
&
\pi\act[a]r=&[\pi(a)]\pi\act r
\\[2ex]
\fv(a) =& \varnothing 
&
\fv(\tf f(r)) =& \fv(r) 
\\
\fv(\pi\act C) =& \varnothing 
&
\fv((r_1,\ldots,r_n)) =& \bigcup_{1\leq i\leq n} \fv(r_i) 
\\
\fv(\pi\act X) =& \{X\} 
&
\fv([a]r) =& \fv(r)
\end{tab2}
\end{defn}

\begin{rmrk}
In Definition~\ref{defn.fa} we in effect give every unknown permission set $\atomsdown$ (so that $\fa(\pi\act X)=\pi\act\atomsdown$).
We obtain the effect of an unknown with permission set $\pi\act\atomsdown$ just by writing $\pi\act X$. 
This simplified design makes Proposition~\ref{prop.commute.pi} easier to express. 
It corresponds roughly to \cite[Example~3.1.7(2)]{gabbay:nomtnl}.
\end{rmrk}

\begin{lemm}
\label{lemm.pi.ftma}
$\fa(\pi\act r)=\pi\act\fa(r)$.
\end{lemm}

\begin{lemm}
\label{lemm.fa.pi.r}
If $\pi(a)=\pi'(a)$ for all $a\in\fa(r)$ then $\pi\act r=\pi'\act r$.
\end{lemm}

\subsection{$\alpha$-equivalence}

\begin{defn}
A \deffont{congruence} is an equivalence relation $R$ such that if $r\mathrel{R} s$ then $\tf f(r)\mathrel{R} \tf f(s)$ and $(t_1,\ldots,r,\ldots,t_n)\mathrel{R}(t_1,\ldots,s,\ldots,t_n)$ and $[a]r\mathrel{R}[a]s$.

\deffont{$\alpha$-equivalence} is then the least congruence such that if $a,b\not\in\fa(r)$ then $(b\ a)\act r\aeq r$.\footnote{This characterisation, which follows \cite{gabbay:forcie}, captures in slightly abstract form three more syntax-directed rules: $b\not\in\fa(r)$ then $[b](b\ a)\act r\aeq [a]r$,\ and if $\pi|_{\pmss{X}}=\pi'|_{\pmss{X}}$ then $\pi\act X\aeq\pi'\act X$,\ and if $\pi|_{\pmssC(C)}=\pi'|_{\pmssC(C)}$ then $\pi\act C\aeq\pi'\act C$.}
\end{defn} 

We do not quotient terms by $\alpha$-equivalence.
The syntax $[a]r$ is a formal pair of $a$ and $r$.
So for example, $[a]X$ and $[b](b\ a)\act X$ for $b\not\in\pmss{X}$ are different concrete terms.

In fact, we never use $\alpha$-equivalence $\aeq$ directly in this paper (it would be needed if we proved soundness and completeness, but these proofs are in other papers and are not included here). 
However $\aeq$ lurks in the background, hard-wired into the denotation: it can be proved that if $r\aeq s$ then $r$ and $s$ will always denote the same element in Definition~\ref{defn.interpret.terms}. 

\subsection{Interpretation of signatures and terms} 

\begin{defn}
\label{defn.f.equivar}
Suppose $\ns X$ and $\ns Y$ are permissive-nominal sets and $F\in |\ns X|\to|\ns Y|$ is a function.
Call $F$ \deffont{equivariant} when $F(\pi\act x)=\pi\act F(x)$ for all permutations $\pi\in\mathbb P_{\fini}$ and $x\in|\ns X|$.
\end{defn}

\begin{defn}
\label{defn.interpretation}
Suppose $(\mathcal A,\mathcal B)$ is a sort-signature (Definition~\ref{defn.sort.sig}).

\begin{frametxt}
A \deffont{interpretation} $\interp I$ for $(\mathcal A,\mathcal B)$ consists of an assignment of a permissive-nominal set $\idenot{}{\alpha}$ to each sort $\alpha$ in $(\mathcal A,\mathcal B)$, along with equivariant maps
\begin{itemize*}
\item
for each $\nu\in\mathcal A$ an equivariant and injective map $\mathbb A_\nu\to\idenot{}{\nu}$ which we write $a^\iden$, 
\item
for each $\nu\in\mathcal A$ and $\alpha$ an equivariant and injective map $[\mathbb A_\nu]\idenot{}{\alpha}\to\idenot{}{[\nu]\alpha}$ which we write $[a]^\iden x$, and
\item
for each $\alpha_i$ for $1\leq i\leq n$ an equivariant and injective map $\Pi_i\idenot{}{\alpha_i}\to \idenot{}{(\alpha_1,\dots,\alpha_n)}$ which we write $(x_1,\ldots,x_n)^\iden$.
\end{itemize*}
\end{frametxt}
\end{defn}

\begin{defn}
\label{defn.Sigma.interpretation}
Suppose $\Sigma=(\mathcal A,\mathcal B,\mathcal C,\mathcal F,\f{ar},\pmssC)$ is a signature (Definition~\ref{defn.signature}).

\begin{frametxt}
A \deffont{($\Sigma$-)interpretation} $\interp I$ for $\Sigma$, or \deffont{$\Sigma$-algebra}, consists of the following data:
\begin{itemize*}
\item
An interpretation for the sort-signature $(\mathcal A,\mathcal B)$ (Definition~\ref{defn.interpretation}).
\item
For every $\tf f\in\mathcal F$ with $\f{ar}(\tf f)=(\alpha)\tau$ an equivariant function $\tf f^\iden$ from $\idenot{}{\alpha}$ to $\idenot{}{\tau}$. 
\item
An assignment of a $C^\iden\in\idenot{}{\sort(\constant)}$ to $C\in\mathcal C$, such that $\supp(C^\iden)\subseteq\pmssC(C)$.
\end{itemize*}
\end{frametxt}
\end{defn}

\begin{defn}
\label{defn.valuation}
Suppose $\interp I$ is a $\Sigma$-algebra. 
A \deffont{valuation} $\varsigma$ to $\interp I$ is an equivariant function on unknowns $\mathcal X$ such that for each unknown $X$,\ 
$\varsigma(X)\in\idenot{}{\sort(X)}$.

$\varsigma$ will range over valuations.
\end{defn}

\begin{defn}
\label{defn.interpret.terms}
Suppose $\interp I$ is a $\Sigma$-algebra. 
Suppose $\varsigma$ is a valuation to $\interp I$.

Extend $\interp I$ to an \deffont{interpretation} on terms $\idenot{\varsigma}{r}$ (where of course $r$ is a term in the signature $\Sigma$) by:
\begin{frameqn}
\begin{noquotetab2}
\idenot{\varsigma}{a} =& a^\iden
&
\idenot{\varsigma}{\tf f(r)} =& 
\tf f^\iden(\idenot{\varsigma}{r})
\\
\idenot{\varsigma}{\constant} =& \constant^\iden 
&
\idenot{\varsigma}{(r_1,\ldots,r_n)} =& 
(\idenot{\varsigma}{r_1},\ldots,\idenot{\varsigma}{r_n})^\iden
\\
\idenot{\varsigma}{\pi\act X} =& \pi\act\varsigma(X)
&
\idenot{\varsigma}{[a]r} =& [a]^\iden\idenot{\varsigma}{r}
\end{noquotetab2}
\end{frameqn}
\end{defn}

Lemmas~\ref{lemm.sort.r} to~\ref{lemm.supp.r} are proved by routine inductions:
\begin{lemm}
\label{lemm.sort.r}
If $r:\alpha$ then $\idenot{\varsigma}{r}\in\idenot{}{\alpha}$.
\end{lemm}

\begin{lemm}
\label{lemm.fv.varsigma}
If $\varsigma(X)=\varsigma'(X)$ for every $X\in\fv(r)$ then $\idenot{\varsigma}{r}=\idenot{\varsigma'}{r}$.
\end{lemm}

\begin{lemm}
\label{lemm.pi.r.model}
$\pi\act\idenot{\varsigma}{r} = \idenot{\varsigma}{\pi\act r}$.
\end{lemm}

\begin{lemm}
\label{lemm.supp.r}
$\supp(\idenot{\varsigma}{r})\subseteq\f{fa}(r)$.
\end{lemm}

Looking ahead, later on in Section~\ref{sect.completeness.fin}, we use interpretations to define a notion of validity with respect to a model or a collection of models, written $\interp H\ment r=s$ and $\theory T\ment r=s$.

\section{Reducing support of an interpretation}
\label{sect.models.with.finite.support}

In this section we show how, given an interpretation $\interp H$, to build an interpretation $[m]\interp H$ with `smaller' support.

$[m]\interp H$ will have `almost the same structure' as $\interp H$.
If two terms have a distinct denotation in $\interp H$ then their interpretation in $[m]\interp H$ is also distinct (Proposition~\ref{prop.FT}, which is essentially Theorem~\ref{thrm.pull.out.l} combined with Lemma~\ref{lemm.equality.of.abstractions.X}).

As we shall see in Section~\ref{sect.completeness.fin}, this result can be leveraged to proofs of completeness with respect to interpretations with finite support, assuming completeness with respect to all interpretations.

The idea of the construction is simple: in Definition~\ref{defn.abstraction.X} we take $\interp H$ and abstract all but finitely many atoms in its elements---in Definition~\ref{defn.F} we show how to combine this with the interpretation of the term-formers of $\interp H$.

One way to think of this, is that we replace atoms by numerical indexes (where $a$ is identified with its position in the infinite list of abstractions which we impose).
We can think of $[m]\interp H$ as an abstract `de Bruijn indexes' version of $\interp H$, 
where we recall that de Bruijn indexes are a method of representing object-level variables as numerical indexes \cite{bruijn:lamcnn} typically applied concretely to formal syntax rather than to models.
More on this in the Conclusions.

\subsection{Abstraction by atoms and by infinite lists of distinct atoms: $[a]x$ and $[l]x$}
\label{subsect.abs.X}

\begin{defn}
\label{defn.L}
Choose a fixed but arbitrary enumeration $\lstel{\mn 1},\lstel{\mn 2},\lstel{\mn 3},\ldots$ of some subset of $\atomsdown$---since atoms are countable, this can be done.
Write this enumeration as a list, $l_\ast=[\lstel{\mn 1},\lstel{\mn 2},\lstel{\mn 3},\ldots]$.\footnote{We use negative indexes because we wrote $\atomsdown$ with a $<$.  Of course this does not matter, but it does allow the diagram in Section~\ref{sect.shift} to make geometric sense.}

Define a permissive-nominal set $\mathbb L$ (parameterised by $l_\ast$) by:
$$
\begin{array}{r@{\ }l}
\pi\act l_\ast=&[\pi(\lstel{\mn 1}),\pi(\lstel{\mn 2}),\pi(\lstel{\mn 3}),\ldots]\qquad\qquad\qquad\quad\ \,  
\\
|\mathbb L|=&\{\pi\act l_\ast\mid \text{all }\pi\}
\end{array}
$$
$l$ will range over elements of $|\mathbb L|$.
\end{defn}

It is very easy to check that $\mathbb L$ is indeed a permissive-nominal set, and that $\supp(l)$ is equal to the atoms in $l$.

We will be most interested in the cases of Definition~\ref{defn.L} when $l_\ast$ enumerates all of $\atomsdown$ (Section~\ref{sect.completeness.fin}) and when $l_\ast$ enumerates `half' of $\atomsdown$ (Section~\ref{sect.pnl}).
However, nothing in the mathematics below will depend on this.

\begin{defn}
If $A\subseteq\mathbb A$ define $\fix(A)$ by:
$$
\fix(A)=\{\pi \mid \Forall{a{\in}A}\pi(a)=a\}
$$
\end{defn}

\begin{defn}
\label{defn.abstraction.X}
Suppose $\ns X$ is a permissive-nominal set and $x\in|\ns X|$. 
Suppose $l\in|\mathbb L|$.
Define 
$[l]x$ and $[\mathbb L]\ns X$ as follows: 
\begin{frameqn}
\begin{array}{r@{\ }l}
[l]x =& \{(\pi\act l,\pi\act x) \mid \pi\in\f{fix}(\f{supp}(x){\setminus} \f{supp}(l))\}
\\
|[\mathbb L]\ns X|=&\{[l]x \mid x\in|\ns X|,\ l\in|\mathbb L|\}
\\
\pi\act[l]x =& [\pi\act l]\pi\act x 
\end{array}
\end{frameqn}
\end{defn}

\begin{rmrk}
$[l]x$ and $[\mathbb L]\ns X$ mirror $[a]x$ and $[\mathbb A]\ns X$ from Definition~\ref{defn.abstraction.sets}, and have broadly similar properties.
The idea of abstracting over infinitely many atoms was investigated in \cite{gabbay:genmn} (see equation (2) in Subsection~2.1).
\end{rmrk}

\begin{lemm}
\label{lemm.equality.of.abstractions.X}
Suppose $\ns X$ is a permissive-nominal set and $x,y\in|\ns X|$.
Suppose $l\in|\mathbb L|$.

Then $[l]x=[l]y$ if and only if $x=y$.
\end{lemm}
\begin{proof}
Clearly if $x=y$ then $[l]x=[l]y$.
Suppose $[l]x=[l]y$.
By construction $(l,x)\in [l]x$, so also $(l,x)\in [l]y$.
It follows that there exists $\pi$ such that $\pi\act l=l$ and $\pi\act y=x$, and $\pi\in\fix(\supp(y)\setminus\supp(l))$.
From $\pi\act l=l$ follows that $\pi\in\fix(\supp(l))$.
It follows that $\pi\in\fix(\supp(y))$ and so by Lemma~\ref{lemm.supp.restricted} that $\pi\act y=y$.
\end{proof}

\begin{lemm}
\label{lemm.abstraction.support.x}
Suppose $\ns X$ is a permissive-nominal set and $x\in|\ns X|$.
Suppose $l\in|\mathbb L|$.

Then $\supp([l]x) = \supp(x) {\setminus} \supp(l)$.
\end{lemm}
\begin{proof}
By properties of the group action if $\pi\in\fix(\supp(x){\setminus}\supp(l))$ then $\pi\act [l]x=[\pi\act l]\pi\act x$.

Now suppose $a\in\supp(x){\setminus}\supp(l)$ and choose any $b$ fresh (so $b\not\in\supp(x)\cup\supp(l)$).
It is easy to use Lemma~\ref{lemm.supp.pi.x} to verify that every $(l',x')\in[l]x$ satisfies $a\in\supp(x')$ whereas every $(l',x')\in (b\ a)\act [l]x$ satisfies $a\not\in\supp(x')$.
It follows that $(b\ a)\act [l]x\neq [l]x$ and so by Corollary~\ref{corr.notinsupp} $a\in \supp([l]x)$.
\end{proof}

\begin{corr}
\label{corr.LX.pns}
$[\mathbb L]\ns X$ from Definition~\ref{defn.abstraction.X} \emph{is} a permissive-nominal set.
\end{corr}
\begin{proof}
That it is a set with a permutation action is clear.
That every element has a supporting permission set follows from Lemma~\ref{lemm.abstraction.support.x}.
\end{proof}

\begin{lemm}
\label{lemm.at.l}
Suppose $\ns X$ is a permissive-nominal set. 
Suppose $\hat x\in|[\mathbb L]\ns X|$ and $l\in|\mathbb L|$ is such that $\supp(\hat x)\cap\supp(l)=\varnothing$.
Then there exists a unique element, write it $\hat x\at l\in|\ns X|$, such that $\hat x=[l](\hat x\at l)$.
\end{lemm}
\begin{proof}
By Lemma~\ref{lemm.equality.of.abstractions.X} $\hat x\at l$ is unique if it exists.

Suppose $\supp(l)\cap\supp(\hat x)=\varnothing$.
By construction (Definition~\ref{defn.abstraction.X}) $\hat x=[l']x'$ for some $l'\in\mathbb L$ and $x'\in|\ns X|$.
By construction (Definition~\ref{defn.L}) $l'=\pi\act l$ for some $\pi$.\footnote{This is the crux of the proof: $\mathbb L$ is composed of a single orbit under the permutation action.}
It is also a fact that since $\supp(l)\cap\supp(\hat x)=\varnothing$ and (by Lemma~\ref{lemm.abstraction.support.x}) $\supp(l')\cap\supp(\hat x)=\varnothing$, we can suppose without loss of generality that $\nontriv(\pi)\cap\supp(\hat x)=\varnothing$.
It follows that $\hat x=[l]\pi^\mone\act x'$ and so $\hat x\at l$ exists and is equal to $\pi^\mone\act x'$.
\end{proof}

\begin{lemm}
\label{lemm.factor.out}
Suppose $y_1,\ldots,y_n\in |[\mathbb L]\ns X|$.
Then for any $l$ such that $\supp(l)\cap\bigcup \supp(y_i)=\varnothing$,\ there exist $x_1,\ldots,x_n\in|\ns X|$ such that $y_i=[l]x_i$ for $1\leq i\leq n$.
\end{lemm}
\begin{proof}
We use Lemma~\ref{lemm.at.l} and take $x_i=y_i\at l$.
\end{proof}

\subsection{Restricting permutations $\pi/S$}
\label{subsect.pi/S}

Intuitively, $\pi/S$ (Definition~\ref{defn.pi.S}) is the `smallest' permutation to agree with $\pi$ on $S$.
$\pi/S$ is `trying' to be $\pi|_S$ (Definition~\ref{defn.restrict}) 
but $\pi/S$ is a total function and furthermore is a permutation.
The main result is Theorem~\ref{thrm.leq.least}, and we use $\pi/S$ in Theorem~\ref{thrm.pull.out.l}.

As nominal techniques demonstrate, permutations are an attractive way to handle name-binding.
Think of $\pi/S$ as a version of $\pi|_S$ that we can use if we want to stay in the world of permutations.

\begin{xmpl}
\label{xmpl.pi/S}
Suppose $\pi=(a\;b\;c\;d\;e)(f\;g)$ (so $\pi$ maps $a$ to $b$ to $c$ to $d$ to $e$ to $a$, and $f$ to $g$ to $f$). 
Then:
$$
\begin{array}{@{\ \ \,}r@{\ }l@{\qquad\qquad\,}r@{\ }l}
\pi/\{a\}=&(a\ b\ e) 
&
\pi/\{a,b\}=&(a\ b\ c\ e)
\\
\pi/\{a,c\}=&(a\ b\ c\ d\ e)
&
\pi/\{a,f\}=&(a\ b\ e)(f\ g)
\end{array}
$$
Suppose $\pi=(a\;b\;c\;d\;e\;f)$.
Then
$$
\begin{array}{r@{\ }l@{\qquad}r@{\ }l}
\pi/\{b,e\} =& (a\;b\;c)(d\;e\;f)
&
\pi/\{b\} =& (a\;b\;c)
\\
\pi/\{b,e,d\} =& (a\;b\;c\;d\;e\;f)
&
\pi/\{a,d\} =& (a\;b\;f)(c\;d\;e)
\end{array}
$$
\end{xmpl}

Recall the definitions of $\nontriv(\pi)$ and $\pi$ from Definition~\ref{def.nontriv}. 
\begin{defn}
\label{defn.pi.S}
Represent permutations $\pi$ as cycles; so we write $\pi$ as a finite set of finite cycles indexed by $i\in I$ where cycle number $i$ has length $\alpha_i>1$: 
$$
\pi = \Pi_{i\in I} (a_{i1}\ a_{i2}\ \ldots\ a_{i\alpha_i}) 
$$
Define $\pi/S$ as that permutation obtained as follows: 
\begin{frametxt}
\begin{itemize*}
\item
Delete from the cycle representation of $\pi$ above any atom $a$ such that $\{a,\pi(a),\pi^\mone(a)\}\cap S=\varnothing$. 
That is, if there is any part of a cycle of the form `$a_1\ a_2\ a_3$' where $a_1\not\in S$, $a_2\not\in S$, and $a_3\not\in S$, then we replace it with `$a_1\ a_3$'.
Repeat, until we cannot proceed.
\item
If there is any part of a cycle of the form `$a_1\ a_2\ a_3\ a_4$' where $a_1\in S$ and $a_4\in S$ but $a_2\not\in S$ and $a_3\not\in S$, break the cycle into two subcycles as follows:\ `$a_1\ a_2)(a_3\ a_4$'. 
\end{itemize*}
\end{frametxt}
\end{defn} 
In words: 
\begin{quote}
$\pi/S$ is obtained from $\pi$ by eliding sequences of three or more consecutive atoms not in $S$, and then by splitting cycles at any two consecutive atoms not in $S$.
\end{quote}

\begin{lemm}
$\pi/S$ is well-defined.
\end{lemm}
\begin{proof}
At each step the size of $\nontriv$ reduces, so the rewrite system is terminating.
It is not hard to check that rewrites are locally confluent.
The result follows by Newman's Lemma \cite{NewmanM:thecde}.
\end{proof}

\begin{defn}
\label{defn.pi|S}
Define $\pi'\leq_S \pi$ when:
\begin{itemize*}
\item
$\pi'|_S=\pi|_S$
\item
$(\pi')^\mone|_S=\pi^\mone|_S$
\item
For every cycle with atoms $C'$ in $\pi'$, there is a cycle with atoms $C$ in $\pi$ such that $C'\subseteq C$.
\end{itemize*}
\end{defn}
It is easy to verify that $\leq_S$ is a transitive reflexive relation.
$\leq_S$ is not antisymmetric: if $\pi=(a\ b\ c)$ and $\pi'=(a\ c\ b)$ and $S=\varnothing$ then $\pi\leq_S \pi'$ and $\pi'\leq_S\pi$ yet $\pi\neq \pi'$.

\begin{thrm}
\label{thrm.leq.least}
\begin{enumerate}
\item
$\pi/S$ is the unique $\leq_S$-least permutation beneath $\pi$.
\item
As a corollary, $(\pi/S)|_S=\pi|_S$ and if $\pi|_S=\pi'|_S$ and $\pi^\mone|_S=(\pi')^\mone|_S$ then $\pi/S=\pi'/S$.
\end{enumerate}
\end{thrm}
\begin{proof}
By construction $\pi/S$ contains only those atoms, in the smallest possible cycles, necessary to agree with $\pi$ and $\pi^\mone$ on $S$.
\end{proof}

\subsection{Making support smaller}

Given an interpretation $\interp H$ and a list of atoms $m$, we are interested in `subtracting' $m$ from the support of $\interp H$, in some sense.
The main definition is Definition~\ref{defn.F}, which builds an interpretation with smaller support out of an interpretation.
For the cases we care about, `smaller support' will mean finite support; this will come later in Lemmas~\ref{lemm.LX.ordinary} and~\ref{lemm.LX.medium.ordinary}, which are then used in Theorems~\ref{thrm.fin.nonfin} and Theorem~\ref{thrm.pnl.fin.nonfin} respectively.
Here, we give the relevant construction.

\begin{defn}
\label{defn.F}
Given a signature $\Sigma$,\ a $\Sigma$-interpretation $\interp H$,\ and a list $m\in|\mathbb L|$ construct a $\Sigma$-interpretation $[m]\interp H$ as follows:
\begin{frameqn}
\begin{array}{@{\hspace{-2em}}r@{\ }l@{\hspace{0em}}r@{\ }l}
\fmdenot{\alpha}=&\{[l]x\mid l\in|\mathbb L|,\ x\in\hdenot{\alpha}\} 
\\[2ex]
a^\fmden =& [l](a^\hden) \ \ (\supp(l)\not\ni a)\hspace{-10em}
&
\tf f^\fmden([l]x) =& 
[l]\tf f^\hden(x)
\\
([l]x_1,\ldots,[l]x_n)^\fmden =& 
[l](x_1,\ldots,x_n)^\hden
&
[a]^\fmden([l]x) =& [l]([a]^\hden x)
\\
C^\fmden =& [m]C^\hden
\end{array}
\end{frameqn}
\end{defn} 

\begin{rmrk}
A couple of comments on Definition~\ref{defn.F}:

The index $m$ of $[m]\interp H$ is only used to interpret constants $C$.
We have to choose \emph{some} list of atoms to abstract---if our language did not admit non-equivariant constants, as was the case for the original Urban-Pitts-Gabbay syntax from \cite{gabbay:nomu-jv} or its permissive variant from e.g. \cite{gabbay:perntu-jv}, then we could just write $[\mathbb L]\interp H$.

In the case of tuples, we know we can write every element in the form $[l]x_i$ for $1\leq i\leq n$ for some $x_i$, by Lemma~\ref{lemm.factor.out}.
\end{rmrk}

\begin{prop}
\label{prop.F.finite.interpretation}
$[m]\interp H$ from Definition~\ref{defn.F} is an interpretation.
\end{prop}
\begin{proof} 
It is routine to check that every condition in Definitions~\ref{defn.interpretation} and~\ref{defn.Sigma.interpretation} is satisfied.
\end{proof}

The next step is to build valuations to $[m]\interp H$.
This is Definition~\ref{defn.l.varsigma} and Proposition~\ref{prop.varsigma.Z.is.a.valuation}.
\begin{defn}
\label{defn.l.varsigma}
Suppose $\varsigma$ is a valuation to $\interp H$ and $l\in|\mathbb L|$. 
Define $[l]\varsigma$ by:
\begin{frameqn}
([l]\varsigma)(X) = [l](\varsigma(X)) 
\end{frameqn}
\end{defn}

\begin{prop}
\label{prop.varsigma.Z.is.a.valuation}
If $\varsigma$ is a valuation to $\interp H$ then $[l]\varsigma$ is a valuation to $[m]\interp H$.
\end{prop}
\begin{proof}
Consider an unknown $X$.
By assumption $\varsigma(X)\in\hdenot{\sort(X)}$ and $\supp(\varsigma(X))\subseteq\pmss{X}$.
By construction in Definitions~\ref{defn.the.comb} and~\ref{defn.abstraction.X},\ $\pmss{X}\setminus\supp(l)$ is finite so by Lemma~\ref{lemm.abstraction.support.x},\ $\supp([l]\varsigma(X))$ is finite.
The result follows.
\end{proof}

\section{Three commutation results}
\label{sect.two.comm}

Theorem~\ref{thrm.pull.out.l},\ Lemma~\ref{lemm.l.varsigma},\ and Proposition~\ref{prop.commute.pi}
are three commutation results.
In Sections~\ref{sect.completeness.fin} and~\ref{sect.pnl} we will use these as the technical `engine' behind main theorems such as Theorems~\ref{thrm.fin.nonfin} and~\ref{thrm.pnl.fin.nonfin}.

\subsection{Atoms of a term}

First, we need a technical tool $\atoms(r)$.
We need this to express the side-condition $\atoms(r)\cap\f{supp}(l)=\varnothing$ in Theorem~\ref{thrm.pull.out.l}, 
and the side-condition $\atoms(r)\cap\nontriv(\pi)=\varnothing$ in Proposition~\ref{prop.commute.pi}.
Without these side-condition, the results would not hold.

\begin{defn}
\label{defn.atomsof}
Define $\atoms(r)$ inductively by:
\begin{frameqn}
\begin{array}{r@{\ }l@{\quad}r@{\ }l}
\atoms(a)=&\{a\}
&
\atoms(\tf f(r))=&\atoms(r)
\\
\atoms(\pi\act C)=&\nontriv(\pi/\pmssC(C))
&
\atoms((r_1,\ldots,r_n))=&\bigcup \atoms(r_i)
\\
\atoms(\pi\act X)=&\nontriv(\pi/\pmss{X})  
&
\atoms([a]r)=& \atoms(r)\cup\{a\} 
\end{array}
\end{frameqn}
\end{defn}
$\atoms(r)$ collects the atoms `explicit' in $r$.
Contrast this with `free atoms of' $\fa(r)$ from Definition~\ref{defn.fa} which collects the atoms `potentially' in $r$.
For instance, $\fa(X)=\pmss{X}$ and is infinite, but $\atoms(X)=\varnothing$.
This is because $X$ mentions no atoms explicitly, but intuitively it could be instantiated for any term with atoms in $\pmss{X}$.

\subsection{First commutation result}

Recall from Definition~\ref{defn.L} the construction of $\mathbb L$, parameterised over some $l_\ast$.
\begin{thrm}
\label{thrm.pull.out.l}
Suppose $l\in|\mathbb L|$
and $\atoms(r)\cap \f{supp}(l)=\varnothing$. 
Then $\denot{[l]\interp H}{[l]\varsigma}{r} = [l]\denot{\interp H}{\varsigma}{r}$.
\end{thrm}
\begin{proof}
By induction on $r$:
\begin{itemize}
\item
\emph{The case $a$.}\quad
We reason as follows:
\maketab{tab1}{@{\hspace{-2em}}R{8em}@{\ }L{10em}L{15em}}
\begin{tab1}
\denot{[l]\interp H}{[l]\varsigma}{a}
=&
[l]a^{\hden}
&\text{Defs~\ref{defn.interpret.terms},~\ref{defn.F}},\ a\not\in\f{supp}(l) 
\\
=&
[l]\denot{\interp H}{\varsigma}{a}
&\text{Definition~\ref{defn.interpret.terms}}
\end{tab1}
We know $a\not\in\f{supp}(l)$ because we assumed $\atoms(r)\cap\f{supp}(l)=\varnothing$, and $\atoms(a)=\{a\}$.
\item
\emph{The case $\pi\act X$.}\quad
We reason as follows:
\begin{tab1}
\denot{[l]\interp H}{[l]\varsigma}{\pi\act X}=&
\pi\act [l]\varsigma(X)
&\text{Definition~\ref{defn.interpret.terms}}
\\
=&(\pi/\pmss{X})\act[l]\varsigma(X)
&
\text{Lems~\ref{lemm.supp.restricted} \&~\ref{lemm.abstraction.support.x},\ Thm~\ref{thrm.leq.least}}
\\
=&[l](\pi/\pmss{X})\act\varsigma(X)
&\text{Fact} 
\\
=&[l]\pi\act\varsigma(X)
&
\text{Lems~\ref{lemm.supp.restricted} \&~\ref{lemm.abstraction.support.x},\ Thm~\ref{thrm.leq.least}}
\\
=&[l]\denot{\interp H}{\varsigma}{\pi\act X}
&\text{Definition~\ref{defn.interpret.terms}}
\end{tab1}
The fact above follows since we assumed $\atoms(\pi\act X)\cap\supp(l)=\varnothing$.
\item
\emph{The case $[a]r$, where $a\not\in\supp(l)$.}\quad
We reason as follows:
\begin{tab1}
\denot{[l]\interp H}{[l]\varsigma}{[a]r}
=&
[a]^{\flden}\denot{[l]\interp H}{[l]\varsigma}{r}
&\text{Definition~\ref{defn.interpret.terms}}
\\
=&
[a]^{\flden}[l]\denot{\interp H}{\varsigma}{r}
&\text{ind. hyp.}
\\
=&
[l]([a]^{\hden}\denot{\interp H}{\varsigma}{r})
&\text{Definition~\ref{defn.F}} 
\\
=&
[l]\denot{\interp H}{\varsigma}{[a]r}
&\text{Definition~\ref{defn.interpret.terms}} 
\end{tab1}
\item
\emph{The case $\pi\act C$.}\quad
We reason as follows:
\begin{tab1}
\denot{[l]\interp H}{[l]\varsigma}{\pi\act C}
=& \pi\act [l]C^{\hden}
&\text{Defs~\ref{defn.interpret.terms} \&~\ref{defn.F}}
\\
=&(\pi/\pmssC(C))\act[l]C^{\hden}
&
\text{Lems~\ref{lemm.supp.restricted} \&~\ref{lemm.abstraction.support.x},\ Thm~\ref{thrm.leq.least}}
\\
=&[l](\pi/\pmssC(C))\act C^{\hden}
&\text{Fact}
\\
=&[l]\pi\act C^{\hden}
&
\text{Lems~\ref{lemm.supp.restricted} \&~\ref{lemm.abstraction.support.x},\ Thm~\ref{thrm.leq.least}}
\\
=&[l]\denot{\interp H}{\varsigma}{\pi\act C}
&\text{Definition~\ref{defn.interpret.terms}}
\end{tab1}
The fact above follows since we assumed $\atoms(\pi\act C)\cap\supp(l)=\varnothing$.
\item 
\emph{The case $(r_1,\ldots,r_n)$.}\quad
We reason as follows:
\begin{tab1}
\denot{[l]\interp H}{[l]\varsigma}{(r_1,\ldots,r_n)}
=&
(\denot{[l]\interp H}{[l]\varsigma}{r_1},\ldots,\denot{[l]\interp H}{[l]\varsigma}{r_n})^{\flden}
&\text{Definition~\ref{defn.interpret.terms}} 
\\
=&([l]\denot{\interp H}{\varsigma}{r_1},\ldots,[l]\denot{\interp H}{\varsigma}{r_n})^{\flden}
&\text{ind. hyp.}
\\
=&[l](\denot{\interp H}{\varsigma}{r_1},\ldots,\denot{\interp H}{\varsigma}{r_n})^{\hden}
&\text{Definition~\ref{defn.F}}
\\
=&[l]\denot{\interp H}{\varsigma}{(r_1,\ldots,r_n)}
&\text{Definition~\ref{defn.interpret.terms}}
\end{tab1}
\item
\emph{The case $\tf f(r)$ \dots} is routine.
\qedhere
\end{itemize}
\end{proof}

\subsection{Second commutation result}

\begin{defn}
\label{defn.varsigma.Xx}
Given an interpretation $\interp H$, a valuation $\varsigma$ to $\interp H$, and some $X$ and $x\in\hdenot{\sort(X)}$ with $\supp(x)\subseteq\pmss{X}$, define $\varsigma[X\ssm x]$ by:
$$
\varsigma[X\ssm x](X)=x\qquad
\varsigma[X\ssm x](Y)=\varsigma(Y)
$$ 
\end{defn}

\begin{lemm}
\label{lemm.l.varsigma}
Suppose $\varsigma$, $X$, and $x$ are as in Definition~\ref{defn.varsigma.Xx}.
Suppose $l\in|\mathbb L|$.
Then 
$$([l]\varsigma)[X{\ssm} [l]x]) = [l](\varsigma[X{\ssm} x]).$$
\end{lemm}
\begin{proof}
By routine calculations.
\end{proof}

\subsection{Third commutation result}

\begin{defn}
\label{defn.varsigma.pi}
Suppose $\varsigma$ is a valuation.
Suppose $\pi$ is a permutation such that $\nontriv(\pi)\subseteq\atomsdown$.

Define $\pi\fcomp\varsigma$ by
\begin{frameqn}
(\pi\fcomp\varsigma)(X)=\pi\act\varsigma(X) .
\end{frameqn}
\end{defn}

\begin{prop}
\label{prop.commute.pi}
Suppose $\nontriv(\pi)\subseteq\atomsdown$ and $\atoms(r)\cap\nontriv(\pi)=\varnothing$.
Then $\denot{\interp H}{\pi\fcomp\varsigma}{r}=\pi\act\denot{\interp H}{\varsigma}{r}$.
\end{prop}
\begin{proof}
By a routine induction on $r$ similar to that in Theorem~\ref{thrm.pull.out.l}:
\begin{itemize*}
\item
\emph{The case $a$.}\quad
By assumption $a\not\in\nontriv(\pi)$.
\item
\emph{The case $\pi'\act X$.}\quad
By assumption $\nontriv(\pi)\cap\nontriv(\pi'/\pmss{X})=\varnothing$.
Since $\nontriv(\pi)\subseteq\atomsdown$ it is a fact that $\nontriv(\pi)\cap\nontriv(\pi')=\varnothing$.
The result follows. 
\item
\emph{The case $[a]r$, where $a\not\in\supp(l)$.}\quad
By assumption $a\not\in\nontriv(\pi)$.
\item
\emph{The case $\pi'\act C$.}\quad
As for $\pi'\act X$.
\item 
\emph{The cases $(r_1,\ldots,r_n)$ and $\tf f(r)$ \dots} are routine.
\qedhere
\end{itemize*}
\end{proof}

\section{Nominal algebra completeness relative to interpretations with finite support}
\label{sect.completeness.fin}

We now have everything we need to set up two notions of validity $\ment$ and $\mentfin$ (Definition~\ref{defn.M.ment.fin}) and prove our main result, that they are equal (Theorem~\ref{thrm.fin.nonfin}).

\begin{defn}
Suppose $r$ and $s$ are terms in $\Sigma$, which is the signature of an interpretation $\interp H$.
\begin{frametxt}
\begin{itemize*}
\item
Write $\interp H,\varsigma\ment r=s$ when $\denot{\interp H}{\varsigma}{r}=\denot{\interp H}{\varsigma}{s}$.
\item
Write $\interp H\ment r=s$ when $\interp H,\varsigma\ment r=s$ for every valuation $\varsigma$ to $\interp H$. 
\end{itemize*}
\end{frametxt}
\end{defn}

\begin{nttn}
\label{nttn.last.permission.set}
For the rest of this section, we will take $l_\ast$ from Definition~\ref{defn.L} to enumerate all of $\atomsdown$.
We write the $\mathbb L$ so generated by Definition~\ref{defn.L} as $\Ld$. 
\end{nttn}

Recall the construction of $[m]\interp H$ from Definition~\ref{defn.F}.
\begin{prop}
\label{prop.FT}
Suppose $r$ and $s$ are terms in $\Sigma$, which is the signature of an interpretation $\interp H$.
Suppose $m\in|\Ld|$.
Then:
\begin{enumerate*}
\item
If $\interp H\not\ment r=s$ then $[m]\interp H\not\ment r=s$.
\item
If $\interp H\ment r=s$ then $[m]\interp H\ment r=s$.
\end{enumerate*}
\end{prop}
\begin{proof}
For the first part, suppose $\interp H\not\ment r=s$.
So there exists a valuation $\varsigma$ to $\interp H$ such that $\denot{\interp H}{\varsigma}{r}\neq\denot{\interp H}{\varsigma}{s}$.
Choose some $l$ such that $\supp(l)\cap (\atoms(r)\cup\atoms(s))=\varnothing$.
We can do this, because $\atoms(r)$ and $\atoms(s)$ are finite. 
By Theorem~\ref{thrm.pull.out.l} 
$\denot{[m]\interp H}{[l]\varsigma}{r}=[l]\denot{\interp H}{\varsigma}{r}$ and
$\denot{[m]\interp H}{[l]\varsigma}{s}=[l]\denot{\interp H}{\varsigma}{s}$.
By Lemma~\ref{lemm.equality.of.abstractions.X} 
$[l]\denot{\interp H}{\varsigma}{r}\neq [l]\denot{\interp H}{\varsigma}{s}$.
It follows that $\denot{[m]\interp H}{[l]\varsigma}{r}\neq\denot{[m]\interp H}{[l]\varsigma}{s}$.

For the second part, suppose that $\interp H\ment r=s$ and suppose $\varsigma'$ is a valuation to $[m]\interp H$.
Choose some $l\in|\Ld|$ such that 
$$
\supp(l)\cap\Bigl(\atoms(r)\cup\atoms(s)\cup\bigcup\{\supp(\varsigma'(X))\mid X\in\fv(r)\cup\fv(s)\}\Bigr) = \varnothing .
$$
We can do this since all the sets on the right-hand side of $\cap$ are finite.

Using Lemmas~\ref{lemm.factor.out} and~\ref{lemm.fv.varsigma} there exists a valuation $\varsigma$ to $\interp H$ such that 
$\denot{[m]\interp H}{\varsigma'}{r}=\denot{[m]\interp H}{[l]\varsigma}{r}$ 
and 
$\denot{[m]\interp H}{\varsigma'}{s}=\denot{[m]\interp H}{[l]\varsigma}{s}$.
We now reason using Theorem~\ref{thrm.pull.out.l} and Lemma~\ref{lemm.equality.of.abstractions.X}, as in the first part.
\end{proof}

The model $[m]\interp H$ is composed of ordinary---i.e. finitely-supported---nominal sets, in the sense of \cite{gabbay:newaas-jv}:
\begin{lemm}
\label{lemm.LX.ordinary}
Every $[l]x\in\fmdenot{\alpha}$ has finite support.
\end{lemm}
\begin{proof}
It suffices to observe Lemma~\ref{lemm.abstraction.support.x} and note that by assumption $\supp(x)$ is contained in a permission set, and by assumption in Notation~\ref{nttn.last.permission.set} $\supp(l)$ is a permission set, and by construction permission sets differ finitely from one another.
\end{proof}

Definition~\ref{defn.pernat} is standard, e.g. from \cite{gabbay:nomuae} (nominal) or \cite{gabbay:nomtnl} (permissive-nominal):
\begin{defn}
\label{defn.pernat}
A \deffont{(permissive-)nominal algebra theory} $\theory T=(\Sigma,\f{Ax})$ is a pair of a signature $\Sigma$ and a set of equality axioms $\f{Ax}$.  
(So elements of $\f{Ax}$ are pairs $r=s$.)

Suppose $\interp H$ is a $\Sigma$-interpretation (Definition~\ref{defn.Sigma.interpretation}).
Write $\interp H\ment\theory T$ to mean that for every valuation $\varsigma$ to $\interp H$ and every $(r=s)\in\f{Ax}$,\ $\denot{\interp H}{\varsigma}{r}=\denot{\interp H}{\varsigma}{s}$.
\end{defn}

\begin{frametxt}
\begin{defn}
\label{defn.interp.finsupp}
Suppose $\Sigma$ is a signature and $\interp F$ is a $\Sigma$-interpretation.
Say that $\interp F$ has \deffont{finite support} when for every sort $\alpha$ in $\Sigma$ and every $x\in|\fdenot{\alpha}|$, it is the case that $\supp(x)$ is finite.
\end{defn} 
\end{frametxt}

\begin{defn}
\label{defn.M.ment.fin}
Suppose $\theory T=(\Sigma,\f{Ax})$ is a theory.
Then:
\begin{itemize*}
\item
Define ${\theory T\mentfin r=s}$ to mean that $\interp F\ment \theory T$ implies $\interp F\ment r=s$, for every $\Sigma$-interpretation $\interp F$ with finite support. 
\item
Define ${\theory T\ment r=s}$ to mean that $\interp H\ment \theory T$ implies $\interp H\ment r=s$, for every $\Sigma$-interpretation $\interp H$. 
\end{itemize*}
\end{defn}

\begin{thrm}
\label{thrm.fin.nonfin}
Suppose that $\theory T$ is a $\Sigma$-theory.

Then $\theory T\mentfin r=s$ if and only if $\theory T\ment r=s$.
\end{thrm}
\begin{proof}
The right-to-left implication is immediate since an interpretation with finite support is an interpretation.

For the left-to-right implication we prove the contrapositive.
Suppose $\theory T\not\ment r=s$.
So there is an interpretation $\interp H$ such that $\interp H\ment\theory T$ and a valuation $\varsigma$ to $\interp H$ such that $\denot{\interp H}{\varsigma}{r}\neq\denot{\interp H}{\varsigma}{s}$.

Choose any $m\in|\Ld|$.
By part~2 of Proposition~\ref{prop.FT} $[m]\interp H\ment \theory T$.
By part~1 of Proposition~\ref{prop.FT} $[m]\interp H\not\ment r=s$, and by Lemma~\ref{lemm.LX.ordinary} we are done.
\end{proof}

Permissive-nominal algebra is sound and complete with respect to permissive-nominal models (the proof is by a Herbrand construction; see \cite[Subsection~7.5]{gabbay:nomtnl}).
So the relevance of Theorem~\ref{thrm.fin.nonfin} is to give completeness also with respect to interpretations with finite support.

\section{Permissive-nominal logic}
\label{sect.pnl}

Permissive-nominal logic (\deffont{PNL}) extends signatures with \deffont{proposition-formers} $\tf P$ with arity $\alpha$.  
It is `first-order logic over (permissive-)nominal terms'.

Full details can be found in \cite{gabbay:pernl,gabbay:pernl-jv} or \cite[Section~9]{gabbay:nomtnl}. Here, we only give the necessary outline.

\subsection{Sketch of permissive-nominal logic}

\begin{defn}
PNL \deffont{propositions} are defined by
\begin{frameqn}
\phi,\psi ::= \bot \mid \phi\limp\phi \mid \Forall{X}\phi \mid \tf P(r) 
\end{frameqn}
where we insist that $r:\alpha$ (where $\alpha$ is the arity of $\tf P$).
\end{defn}

\begin{defn}
\label{defn.u.equivariant}
if $\ns X$ is a nominal set and $U\subseteq|\ns X|$ call $U$ \deffont{equivariant} when $x\in U\liff \pi\act x\in U$ for all $x\in|\ns X|$ and all permutations $\pi$.\footnote{This notion of equivariance coincides with that of Definition~\ref{defn.f.equivar}, if we consider $U$ as a function to truth-values $\{\top,\bot\}$, such that $\pi\act \top=\top$ and $\pi\act\bot=\bot$ for all $\pi$.} 
\end{defn}

Definition~\ref{defn.phi.interp} corresponds to e.g. \cite[Definition~5.11]{gabbay:pernl}.
\begin{defn}
\label{defn.phi.interp}
An interpretation $\interp H$ maps a term to an element of a permissive-nominal set as in Definition~\ref{defn.interpret.terms}, and maps each $\tf P$ to an equivariant subset $\tf P^\hden\subseteq\hdenot{\alpha}$.

This extends to propositions $\phi$ just as in first-order logic where $\denot{\interp H}{\varsigma}{\phi}$ is a truth-value $\top$ or $\bot$, as follows:
\begin{itemize*}
\item
$\denot{\interp H}{\varsigma}{\bot}$ (the syntax) is equal to $\bot$ (the truth-value).
\item
The PNL of \cite{gabbay:pernl,gabbay:pernl-jv,gabbay:nomtnl} is classical, so 
$\denot{\interp H}{\varsigma}{\phi\limp\psi}$ is interpreted as `not $\denot{\interp H}{\varsigma}{\phi}$ or $\denot{\interp H}{\varsigma}{\psi}$'.
\item
$\denot{\interp H}{\varsigma}{\tf P(r)}$ is equal to `$\denot{\interp H}{\varsigma}{r}$ is an element of $\tf P^\hden$'.
\item 
The only non-obvious case is the universal quantifier, which gets a denotation as follows:
$$
\denot{\interp H}{\varsigma}{\Forall{X}\phi}=\bigwedge\{\denot{\interp H}{\varsigma[X\ssm x]}{\phi} \mid x\in\hdenot{\sort(X)},\ \supp(x)\subseteq\pmss{X}\}
$$
This is non-obvious because the $\forall X$ in $\Forall{X}\phi$ quantifies only over $x$ with support in $\pmss{X}$.
More discussion on this in the Conclusions.
\end{itemize*}
\end{defn}

\subsection{Three notions of validity in denotations}

Three distinct notions of validity will interest us.
They are parameterised by `how many atoms' they allow in support.
This is Definition~\ref{defn.medium.valid}; to express it, we need Definition~\ref{defn.medium}. 

\begin{defn}
\label{defn.medium}
For each $i\in\mathbb N$ fix some set $\atomsdd_i\subseteq\atomsdown_i$ such that $\atomsdd_i$ and $\atomsdown_i\setminus\atomsdd_i$ are both infinite.
Write $\atomsdd=\bigcup_i\atomsdd_i$ and:
\begin{itemize*}
\item
Say that $x\in\hdenot{\alpha}$ has \deffont{medium} support when $\supp(x)\subseteq\pi\act\atomsdd$ for some $\pi$.
\item
Say that $\interp H$ has \deffont{medium} support when for every sort $\alpha$ and every $x\in\hdenot{\alpha}$,\ $x$ has medium support.
\end{itemize*}
\end{defn}

\begin{rmrk}
The point of Definition~\ref{defn.medium} is that $x$ with medium support may have infinite support, but this support cannot exhaust the atoms in $\atomsdown$.
But did we not see this already in Definition~\ref{defn.atoms} when we split $\mathbb A$ into $\atomsdown$ and $\atomsup$?
Yes, but PNL has a $\forall X$, so that now (and unlike was the case in permissive-nominal algebra) we have to worry about exhausting all the atoms in $\atomsdown$ within nested quantifiers.
To see this idea made concrete, consider the proof of Proposition~\ref{prop.not.necessarily}. 
\end{rmrk}

\begin{defn}
\label{defn.medium.valid}
\begin{itemize*}
\item
Write $\ment\phi$ to mean that $\interp H,\varsigma\ment\phi$ for every interpretation $\interp H$ and valuation $\varsigma$ to $\interp H$.
\item
Write $\mentint\phi$ to mean that $\interp H,\varsigma\mentint\phi$ for every interpretation $\interp H$ with medium support and valuation $\varsigma$ to $\interp H$.
\item
Write $\mentfin\phi$ to mean that $\interp F,\varsigma\ment\phi$ for every interpretation $\interp F$ with finite support and valuation $\varsigma$ to $\interp F$.
\end{itemize*}
\end{defn}

\begin{prop}
\label{prop.not.necessarily}
$\ment\phi$ implies $\mentint\phi$.
The reverse implication does not necessarily hold.
\end{prop}
\begin{proof}
The first part is immediate since an interpretation with medium support is also an interpretation.

For the second part it suffices to provide a counterexample.
Suppose a base sort $\tau$ and name sort $\nu$ and variables ${X:\tau}$ and ${Y:\nu}$.
Suppose a predicate $\#:(\nu,\tau)$ with intended meaning `is fresh for'/`is not in the support of'.
Consider the formula $\phi=\Forall{X}\Exists{Y}Y\#X$.
Then $\mentfin \phi$ and $\mentint\phi$, but not $\ment\phi$; it might be that $\varsigma(X)=l$ where $l$ lists all atoms in $\atomsdown$, so there exists no atom in $\atomsdown$ (by Definition~\ref{defn.phi.interp}, $Y$ ranges over atoms in $\atomsdown$) that is not in $\supp(l)$.
\end{proof}

The rest of this section is devoted to proving that $\mentint\phi$ if and only if $\mentfin\phi$ (Theorem~\ref{thrm.pnl.fin.nonfin}).
We discuss the relevance of these results in Subsection~\ref{subsect.pnl.relevance}.

\subsection{Finite support denotations from medium support denotations} 
\label{subsect.fin.med}

\begin{nttn}
\label{nttn.last.medium.set}
For the rest of this section, we will take $l_\ast$ from Definition~\ref{defn.L} to enumerate $\atomsdd$.
We write the $\mathbb L$ generated by Definition~\ref{defn.L} as $\Ldd$.
\end{nttn}

\begin{defn}
\label{defn.fmi}
Given a PNL interpretation $\interp H$ with medium support and a list $m\in|\Ldd|$, generate a PNL interpretation $[m]\interp H$ by extending Definition~\ref{defn.F} such that
\begin{frameqn}
\tf P^{\den{[m]\interp H}} = \{[l]x \mid x\in\tf P^\hden,\ l\in|\Ldd| \} .
\end{frameqn}
\end{defn}
Where does the $m$ in $[m]\interp H$ appear on the right-hand side here?  
It does not: $m$ is only used to reduce the support of the interpretations of constant symbols; see Definition~\ref{defn.F}.
PNL only allows equivariant (Definition~\ref{defn.u.equivariant}) interpretation of proposition-formers.
If we considered a flavour of PNL in which proposition-formers could receive non-equivariant interpretation (so that in the syntax we would allow terms of the form $(\pi\act\tf P)(r)$), then Definition~\ref{defn.fmi} would mention $m$ on the right.
This makes no difference to expressivity since we can emulate the effect of a non-equivariant proposition-former using $\tf P(C,r)$.
Our design follows the path of the simplest definitions and proofs.

\begin{lemm}
\label{lemm.LX.medium.ordinary}
Every $[l]x\in\fmdenot{\alpha}$ 
has finite support.
\end{lemm}
\begin{proof}
As for Lemma~\ref{lemm.LX.ordinary}, but now using Notation~\ref{nttn.last.medium.set} and our assumption that $x\in\hdenot{\alpha}$ has medium support. 
\end{proof}

\begin{defn}
\label{defn.pnl.atoms}
Extend $\atoms(r)$ (Definition~\ref{defn.atomsof}) to propositions $\atoms(\phi)$ inductively by:
\begin{frameqn}
\begin{array}{r@{\ }l@{\qquad}r@{\ }l}
\atoms(\bot)=&\varnothing
&
\atoms(\phi\limp\psi)=&\atoms(\phi){\cup}\atoms(\psi)
\\
\atoms(\tf P(r))=&\atoms(r)
&
\atoms(\Forall{X}\phi)=&\atoms(\phi) 
\end{array}
\end{frameqn}
\end{defn}

Lemma~\ref{lemm.pnl.pi} extends Proposition~\ref{prop.commute.pi} to predicates, and is needed for Proposition~\ref{prop.pnl.mentfin.to.ment}.
Recall from Definition~\ref{defn.varsigma.pi} the definition of $\pi\fcomp\varsigma$: 
\begin{lemm}
\label{lemm.pnl.pi}
Suppose $\varsigma$ is a valuation to an interpretation $\interp H$. 
Suppose $\phi$ is a predicate and $\pi$ a permutation such that $\nontriv(\pi)\subseteq\atomsdown$ and $\nontriv(\pi)\cap\atoms(\phi)=\varnothing$.

Then $\denot{\interp H}{\pi\fcomp\varsigma}{\phi}=\denot{\interp H}{\varsigma}{\phi}$.
\end{lemm}
\begin{proof}
By a routine induction on $\phi$.
We consider two cases:
\begin{itemize*}
\item
\emph{The case of $\tf P(r)$.}
By definition $\denot{\interp H}{\pi\fcomp\varsigma}{\tf P(r)}=\top$ if and only if $\denot{\interp H}{\pi\fcomp\varsigma}{r}\in\tf P^\hden$.
By Proposition~\ref{prop.commute.pi} $\denot{\interp H}{\pi\fcomp\varsigma}{r}=\pi\act\denot{\interp H}{\varsigma}{r}$.
By assumption $\tf P^\hden$ is equivariant (Definition~\ref{defn.u.equivariant}).
\item
\emph{The case of $\Forall{X}\phi$.}\quad
$$
\begin{array}{r@{\ }l@{\qquad}l}
\denot{\interp H}{\pi\fcomp\varsigma}{\Forall{X}\phi}
=&\bigwedge\{\denot{\interp H}{(\pi\fcomp\varsigma)[X{\ssm}x]}{\phi}\mid x\in\denot{\interp H}{}{\sort(X)},\ \supp(x)\subseteq\pmss{X}\}
\\
=&\bigwedge\{\denot{\interp H}{(\pi\fcomp\varsigma)[X{\ssm}\pi\act x]}{\phi}\mid x\in\denot{\interp H}{}{\sort(X)},\ \supp(x)\subseteq\pmss{X}\}
\\
=&\bigwedge\{\denot{\interp H}{\pi\fcomp(\varsigma[X{\ssm}x])}{\phi}\mid x\in\denot{\interp H}{}{\sort(X)},\ \supp(x)\subseteq\pmss{X}\}
\\
=&\bigwedge\{\denot{\interp H}{\varsigma[X{\ssm}x]}{\phi}\mid x\in\denot{\interp H}{}{\sort(X)},\ \supp(x)\subseteq\pmss{X}\}
\\
=&\denot{\interp H}{\varsigma}{\Forall{X}\phi}
\end{array}
$$
\end{itemize*}
\end{proof}

\begin{lemm}
\label{lemm.fa.atoms.bound}
$\fa(\phi)\subseteq\atomsdown\cup\atoms(\phi)$ and
$\fa(r)\subseteq\atomsdown\cup\atoms(r)$. 
\end{lemm}
\begin{proof}
By a routine induction on Definitions~\ref{defn.atomsof} and~\ref{defn.pnl.atoms} and by a routine calculation using part~2 of Theorem~\ref{thrm.leq.least} for the base case of $\pi\act X$.
\end{proof}

\begin{prop}
\label{prop.pnl.mentfin.to.ment}
Suppose $\phi$ is a proposition,\ $\interp H$ is a PNL interpretation with medium support (Definition~\ref{defn.medium}),\ and $\varsigma$ is a valuation to $\interp H$.
Suppose $A\subseteq\mathbb A$ is a finite set of atoms such that $\atoms(\phi)\subseteq A$, and suppose $l\in|\Ldd|$ and $\supp(l)\cap A=\varnothing$.

Then
$\interp H,\varsigma\ment \phi$ if and only if 
$[l]\interp H,[l]\varsigma\ment \phi$.
\end{prop}
\begin{proof}
By induction on $\phi$.
We consider a selection of cases:
\begin{itemize}
\item
\emph{The case of $\tf P(r)$.}\quad
We consider the two implications separately.
\\[4pt]
\rulefont{\Leftarrow}\ 
Suppose $[l]\interp H,[l]\varsigma\ment \tf P(r)$. 
This means that $\denot{[l]\interp H}{[l]\varsigma}{r}\in\tf P^\flden$. 
By Theorem~\ref{thrm.pull.out.l} $\denot{[l]\interp H}{[l]\varsigma}{r}=[l]\denot{\interp H}{\varsigma}{r}$ (note that $\atoms(\tf P(r))=\atoms(r)$), and so by Definitions~\ref{defn.fmi} and~\ref{defn.abstraction.X} 
$\pi'\act\denot{\interp H}{\varsigma}{r}\in\tf P^\hden$ for some $\pi'\in\fix(\supp(\denot{\interp H}{\varsigma}{r})\setminus\supp(l))$.
By equivariance of $\tf P^\hden$ it immediately follows that $\denot{\interp H}{\varsigma}{r}\in\tf P^\hden$ and so that $\interp H,\varsigma\ment \tf P(r)$.
\\[4pt]
\rulefont{\Rightarrow}\ 
Now suppose $\interp H,\varsigma\ment \tf P(r)$, so that by definition $\denot{\interp H}{\varsigma}{r}\in\denot{\interp H}{}{\tf P}$.
As in the previous paragraph by Theorem~\ref{thrm.pull.out.l} $[l]\denot{\interp H}{\varsigma}{r}=\denot{[l]\interp H}{[l]\varsigma}{r}$. 
It follows by Definition~\ref{defn.pnl.atoms} that $\denot{[l]\interp H}{[l]\varsigma}{r}\in\tf P^{\flden}$.
\item
\emph{The case of $\Forall{X}\phi$.}\quad
Again we consider the two implications separately: 
\\[4pt]
\rulefont{\Leftarrow}\quad
Suppose $\interp H,\varsigma\not\ment\Forall{X}\phi$.
Unpacking definitions, this means there is some $x\in|\hdenot{\sort(X)}|$ with $\supp(x)\subseteq\pmss{X}$ and $\interp H,\varsigma[X{\ssm} x]\not\ment \phi$.

By inductive hypothesis 
$[l]\interp H,[l](\varsigma[X{\ssm} x])\not\ment\phi$.
We can use Lemma~\ref{lemm.l.varsigma} to write $[l](\varsigma[X{\ssm}x])$ as $([l]\varsigma)[X{\ssm} [l]x]$.
Furthermore, by assumption $\supp(x)\subseteq\atomsdown$ so by Lemma~\ref{lemm.abstraction.support.x} $\supp([l]x)\subseteq\atomsdown\setminus\supp(l)\subseteq\atomsdown$.
It follows by Definition~\ref{defn.phi.interp} that $[l]\interp H,[l]\varsigma\not\ment\Forall{X}\phi$.
\\[4pt]
\rulefont{\Rightarrow}\quad
Suppose $[l]\interp H,[l]\varsigma\not\ment\Forall{X}\phi$.
Unpacking Definition~\ref{defn.phi.interp} this means there are $x'\in|\hdenot{\sort(X)}|$ and $l'\in|\Ldd|$ such that 
$$
\supp([l']x')\subseteq\pmss{X}
\quad\text{and}\quad 
[l]\interp H,([l]\varsigma)[X{\ssm}[l']x']\not\ment\phi. 
$$
If $\supp([l']x')\cap\supp(l)=\varnothing$ then we may use Lemma~\ref{lemm.at.l} and write $[l']x'$ as $[l](([l']x')\at l)$ and deduce by inductive hypothesis that $\interp H,\varsigma\not\ment\phi$.

Otherwise, we choose some $\pi'$ that maps $\supp([l']x')\cap\supp(l)\neq\varnothing$ to a set of atoms in $\pmss{X}$ that is disjoint from $\supp(l)\cup\atoms(\phi)$, and $\pi'$ fixes all other atoms.
This is possible because by construction
$\supp([l']x')$ is finite and $\pmss{X}\setminus(\supp(l){\cup}\atoms(\phi))$ is infinite (recall that $\atomsdown\setminus\atomsdd$ is assumed infinite).
Using Lemma~\ref{lemm.pnl.pi} 
$$
[l]\interp H,([l]\varsigma)[X{\ssm}\pi'\act([l']x')]\not\ment\phi.
$$
We now proceed as in the case where $\supp([l']x') \cap \supp(l) = \varnothing$.
\item
\emph{The case of $\phi\limp\psi$.}\quad
Suppose $\interp H,\varsigma\ment\phi\limp\psi$. 
This means that $\interp H,\varsigma\not\ment \phi$ or $\interp H,\varsigma\ment\psi$.
By inductive hypothesis this is if and only if $[l]\interp H,[l]\varsigma\not\ment\phi$ or $[l]\interp H,[l]\varsigma\ment\psi$.
In either case $[l]\interp H,[l]\varsigma\ment\phi\limp\psi$, and we are done.
\end{itemize}
\end{proof}

\begin{thrm}
\label{thrm.pnl.fin.nonfin}
$\mentfin\phi$ if and only if $\mentint\phi$.
\end{thrm}
\begin{proof}
The right-to-left implication is immediate, just as in Theorem~\ref{thrm.fin.nonfin}.
The left-to-right implication follows using Proposition~\ref{prop.pnl.mentfin.to.ment} and Lemma~\ref{lemm.LX.medium.ordinary}.
\end{proof}

\subsection{Relevance of the theorem}
\label{subsect.pnl.relevance}

It will help to establish some new terminology:
\begin{nttn}
Suppose $\ns X$ is a set with a permutation action and $x\in|\ns X|$.
\begin{itemize*}
\item
Say the element $x\in|\ns X|$ is \deffont{finite-namespace} when $\supp(x)$ is finite.
Similarly say $\ns X$ is finite-namespace when every $x\in|\ns X|$ has finite support.

This is synonymous with $\ns X$ being a nominal set in the sense of \cite{gabbay:newaas-jv}; cf. also Definition~\ref{defn.interp.finsupp}.
\item
Say $x$ is \deffont{$\atomsdd$-namespace} when $\supp(x)\subseteq\pi\act \atomsdd$ for some $\pi$.
Similarly say that $\ns X$ is $\atomsdd$-namespace when every $x\in|\ns X|$ is $\atomsdd$-namespace. 

This is synonymous with \emph{medium} support from Definition~\ref{defn.medium}.
\item
Say $x$ is \deffont{$\atomsdown$-namespace} when $\supp(x)\subseteq\pi\act \atomsdown$ for some $\pi$.
Similarly say that $\ns X$ is \deffont{$\atomsdown$-namespace} when every $x\in|\ns X|$ is $\atomsdd$-namespace.

This is synonymous with $\ns X$ being a permissive-nominal set in the sense of Definition~\ref{defn.nominal.set} or \cite{gabbay:nomtnl}. 
\end{itemize*}
Similarly we will call interpretations finite-namespace, $\atomsdd$-namespace, and $\atomsdown$-namespace in accordance with the support of their underlying sets.
\end{nttn}

The relevance of Theorem~\ref{thrm.pnl.fin.nonfin} is that a PNL predicate is valid over $\atomsdd$-namespace interpretations if and only if it is valid over finite-namespace interpretations.\footnote{For comparison, \emph{nominal logic} does not have this property \cite{pitts:nomlfo-jv}: there are nominal logic predicates that are valid of all finite-namespace interpretations but not of all $\atomsdd$-namespace interpretations (and thus also not valid of all $\atomsdown$-namespace interpretations).  Nominal logic, of course, is a first-order theory; an axiomatisation in first-order logic similar to the axiomatisation of Fraenkel-Mostowski sets from which it is descended. What makes the languages of this paper different is that they are purpose-built using the dedicated new syntax of (permissive-)nominal terms.}

The PNL of \cite{gabbay:pernl,gabbay:pernl-jv,gabbay:nomtnl} has a sequent system giving a notion of logical entailment which is proved sound and complete for $\ment$, that is, for the collection of $\atomsdown$-namespace interpretations.
This differs from the validity $\mentint$, which is validity over $\atomsdd$-namespace interpretations (those with medium support).
This is a more restricted class of models.

Medium support is a new idea to the theory of PNL.
When models are restricted, more statements become valid (usually).
In this case we get a family of theorems, which is exemplified by Proposition~\ref{prop.not.necessarily}.
It remains to devise a complete proof system for PNL over medium support. 

We would not speculate on whether large or medium support is `better'; we suspect that the situation is similar to the intuitionistic/classical question of whether to allow double negation elimination: sometimes we may want it and sometimes we may not.

The value of Theorem~\ref{thrm.pnl.fin.nonfin} is that it tells us that $\atomsdd$ is as small as we need go in exploring validity: restricting models of PNL further to smaller namespaces, and in particular to finite support, will not give us any extra valid propositions.
As we shall argue in the Conclusions, working with sets with infinite support is often easier than working with sets with finite support, so this matters.

And note the obvious: once we carried out our constructions and applied them to permissive-nominal algebra, we could re-use them for permissive-nominal logic with a relatively slight effort of two pages of mathematics in Subsection~\ref{subsect.fin.med}.

\section{More permission sets, more permutations}
\label{sect.shift}

In Definition~\ref{defn.the.comb} we followed \cite{gabbay:nomtnl} and took permission sets to be sets of the form $\pi\act \atomsdown$.
This captures a simple assertion language about the atoms permitted in unknowns.
The results in this paper are sensitive to the expressivity of this language: if we make it slightly more powerful then the results in this paper fail. 

\subsection{More permission sets}

If we follow e.g. \cite{gabbay:perntu-jv} and take permission sets to be sets of the form $(\atomsdown\setminus A)\cup B$ where $A\subseteq\atomsdown$ and $B\subseteq\atomsup$ are finite, then the results in this paper fail.

This genuinely enlarges the set of permission sets (and so makes the assertion language which they represent, more expressive).
For instance, if $b\in\atomsup$ then there is no finite permutation $\pi$ such that $\pi\act\atomsdown=\atomsdown\cup\{b\}$.

\begin{nttn}
Write $\mathcal P$ for the set of all sets of atoms differing finitely from $\atomsdown$ as just described.
\end{nttn}

\begin{prop}
\label{prop.upgrade.zero.fail}
There exists a theory ${\theory T}$ in permissive-nominal algebra with permission sets from ${\mathcal P}$, and an assertion ${r'=s'}$ in that theory, such that ${\theory T\mentfin r'=s'}$ and ${\theory T\not\ment r'=s'}$ (where models are permissive-nominal sets with permission sets in ${\mathcal P}$).
\end{prop}
\begin{proof}
Assume one base type $\tau$ and one term former $0:\tau$ with $\pmssC(0)=\varnothing$.
Assume an axiom $X=0$ where $\fa(X)=\atomsdown$. 
Assume an unknown $Z$ with $\sort(Z)=\tau$ and $\fa(Z)=\atomsdown\cup\{b\}$ where $b\in\atomsup$.
Then:
\begin{itemize}
\item
$\theory T\mentfin Z=0$.

For suppose $\interp F$ is an interpretation of $\theory T$ with finite support: then for any $x\in\denot{\interp F}{}{\tau}$, there is some finite $\pi$ with $\supp(\pi\act x)\subseteq\atomsdown$, hence by our single axiom $\pi\act x = 0$ and by equivariance $x = \pi^\mone \act 0 = 0$, since $0$ has empty support.
\item
However, $\theory T\not\ment Z=0$.

To see this, interpret $\tau$ to be the set $\{\pi\act(\atomsdown\cup\{b\})\mid \pi\text{ finite}\}\cup\{\varnothing\}$,
interpret $0$ by $\varnothing$, and take $\varsigma(Z)=\atomsdown\cup\{b\}$.
\end{itemize}
\end{proof}

Initially we used $\mathcal P$; notably in \cite{gabbay:pernl,gabbay:perntu-jv,gabbay:pernl-jv}.
However, in later papers such as \cite{gabbay:nomtnl} we preferred the design of Definition~\ref{defn.the.comb} because it seemed to make some proofs easier to express.
In the light of the results of this paper we can now better understand the significance of our design choices:
Proposition~\ref{prop.upgrade.zero.fail} suggests that our design in Definition~\ref{defn.the.comb} is mathematically more elementary and somewhat closer to the design `nominal terms + finitely-supported nominal sets' from the previous literature.
That is, the design of Definition~\ref{defn.the.comb} and \cite{gabbay:nomtnl} is the closest `permissive' version of traditional nominal techniques, and the design of \cite{gabbay:pernl,gabbay:perntu-jv,gabbay:pernl-jv} is slightly but measurably more expressive.

\subsection{\emph{shift}-permutations}
\label{subsect.shift}

In the presence of infinite permutations, the results in this paper fail.
We sketch the mathematics involved, starting with a justification of why infinite permutations are an interesting case to consider.

For simplicity assume a single sort of atom.
\begin{defn}
\label{defn.shift}
Suppose $a\in\atomsdown=\{a,a_{\mn 1},a_{\mn 2},a_{\mn 3},\dots\}$ 
and $\atomsup=\{a_1,a_2,a_3,\dots\}$.

Assume a bijection $\shift$ on atoms mapping $\atomsdown$ to $\atomsdown\setminus\{a\}$ and such that $\mathbb A\setminus\nontriv(\pi)$ is infinite (we can do this because we assumed that $\mathbb A$ is countable).
\end{defn}
We illustrate an example:
$$
\text{Illustration of $\shift$:}\ \scalebox{.8}{$\xymatrix@=3ex{
a_{\mn 6} \ar@/^1pc/@{<-}[r] & 
a_{\mn 5} \ar@/^1pc/@{<-}[r] & 
a_{\mn 4} \ar@/^1pc/@{<-}[r] & 
a_{\mn 3} \ar@/^1pc/@{<-}[r] & 
a_{\mn 2} \ar@/^1pc/@{<-}[r] & 
a_{\mn 1} \ar@/^1pc/@{<-}[r] & 
a \ar@/^1pc/@{<-}[rr]& 
a_{1} \ar@(ld,rd)& 
a_{2} \ar@/^1pc/@{<-}[rr]& 
a_{3} \ar@(ld,rd)& 
a_{4} \ar@/^1pc/@{<-}[rr]& 
a_{5} \ar@(ld,rd)&
a_{6} 
}$}
$$
Call $\shift$ a \deffont{shift}-permutation.

$\shift$ has a measurable and favourable effect on the mathematics and algorithmics of nominal syntax.
For instance:
\begin{itemize}
\item
$\shift$ nontrivially increases the deductive power of $\forall X$ in PNL \cite[Subsection~2.7]{gabbay:pernl-jv}.

If $\fa(X)=\atomsdown$ where $a\in\atomsdown$ then $\Forall{X}\tf R(X,X)$ does not entail $\tf R((X,a),(X,a))$ without $\shift$, but it \emph{does} entail $\tf R((X,a),(X,a))$ \emph{with} $\shift$ (for $\tf R$ having an appropriate arity).
This extra power is irrelevant if we only care about finitely-supported models, which is why the issue has not arisen in previous work.
\item
\emph{shift}-permutations can be used to obtain a particularly concise unification algorithm \cite[Section~4]{gabbay:nomtnl}.
\end{itemize}
For more discussion see \cite[Subsection~3.6]{gabbay:nomtnl}.

This extra power is not particularly expensive: we can do what we are used to in nominal techniques, in the presence of $\shift$.
Indeed, the results of \cite{gabbay:nomtnl} are parameterised over a permutation group general enough to admit $\shift$ because this was \emph{easier} than excluding it. In particular the specific design of the nominal unification algorithm and HSP result there, are shorter and simpler because of their use of $\shift$. 

However, in the presence of $\shift$ the results of this paper fail.
Proposition~\ref{prop.upgrade.fail} gives an example of a signature for which permissive-nominal algebra $\ment$ (all permissive-nominal models) is complete, but $\mentfin$ (models with finite support) are not.
In order to state this result we must `upgrade' the material in this paper with $\shift$.
\begin{defn}
To augment Sections~\ref{sect.perns} and~\ref{sect.nominal.terms.syntax} with a shift permutation $\shift$, we proceed as follows:
\begin{enumerate*}
\item
In Definition~\ref{def.permutation} permutations are finitely generated by swappings \emph{and} $\shift$ (they remain finitely representable, but $\nontriv(\pi)$ is now not always finite).

Write $\mathbb P_\shift$ for the group of bijections generated by swappings and $\shift$.
\item
In Definition~\ref{defn.perm.set} assume the permutation action has type $(\mathbb P_{\shift}\times |\ns X|)\to |\ns X|$.
So permissive-nominal sets have an action by swappings and $\shift$.
\item 
In Definition~\ref{defn.support} we say that $A\subseteq\mathbb A$ supports $x\in |\ns X|$ when for every permutation $\pi\in\mathbb P_{\delta}$, if $\pi(a) =a$ for all $a \in A$ then $\pi\act x =x$.\footnote{%
This is a little stronger than we need.
We could also retain the condition that $\pi$ be finite in the definition of supporting set, so we say that $A\subseteq\mathbb A$ supports $x\in |\ns X|$ when for every \emph{finite} permutation $\pi\in\mathbb P_{\fini}$ (so no $\shift$), if $\pi(a) =a$ for all $a \in A$ then $\pi\act x =x$.
We only ever $\alpha$-convert by finitely many atoms in this paper, so the proofs remain unchanged.

What does happen is that we admit models with elements which are fixed by finite permutations, but perhaps not by $\delta$.
For more on this design see \cite{gabbay:pernl-jv}, in particular Remark~3.3.
}
\item
In the examples of Subsection~\ref{subsect.pn.sets.examples} extend for the extra permutations in the natural way.
So $\pi\act a=\pi(a)$ for $\pi\in\mathbb P_\shift$ and $\pi\act[a]x=[\pi(a)]\pi\act x$ for $\pi\in\mathbb P_\shift$.
\item
In Definition~\ref{defn.terms} extend terms also with the extra permutations.
So $\pi\act X$ is a term for $\pi\in\mathbb P_\shift$.
The permutation action Definition~\ref{defn.fa} extends in the natural way.
\item
We extend the notion of equivariance (Definition~\ref{defn.f.equivar}) with the extra permutations.
So $F$ is \emph{equivariant} when $F(\pi\act x)=\pi\act F(x)$ for all permutations $\pi\in\mathbb P_\shift$ and $x\in|\ns X|$.\footnote{We then call the notion of equivariance from Definition~\ref{defn.f.equivar} \deffont{finite equivariance}.   It is possible to be equivariant for finite permutations but not for $\shift$.  The proof of Proposition~\ref{prop.upgrade.fail} will depend on this.}
\end{enumerate*}
\end{defn}

\begin{prop}
\label{prop.upgrade.fail}
There exists a theory ${\theory T}$ in permissive-nominal algebra with ${\shift}$, and an assertion ${r'=s'}$ in that theory, such that ${\theory T\mentfin r'=s'}$ and ${\theory T\not\ment r'=s'}$ (where models are permissive-nominal sets with ${\shift}$).
\end{prop}
\begin{proof}
Assume no term-formers and one base type $\tau$.
Assume $a\in\atomsdown$ and a shift permutation $\shift$ bijecting $\atomsdown$ with $\atomsdown\setminus\{a\}$, as illustrated just after Definition~\ref{defn.shift}.

Assume an axiom $(b\ a)\act X=X$ where 
$b\not\in\atomsdown$.
Then:
\begin{itemize*}
\item
If $\interp F$ is a model of $\theory T$ with finite support then $\supp(x)=\varnothing$ for every $x\in\denot{\interp F}{}{\tau}$.
For suppose there exists $x$ with $\supp(x)\neq\varnothing$.
By equivariance we may (apply a permutation to $x$ to) assume without loss of generality that $a,b\not\in\supp(x)$.

Now choose some $a'\in\supp(x)$ and choose some $\pi$ mapping $\supp(x)$ to a subset of $\atomsdown$ and such that specifically $\pi(a')=a$. 
By our axiom, $(b\ a)\act(\pi\act x) =\pi\act x$.
It follows by calculations on permutations that $(b\ a')\act x=x$ and so by Corollary~\ref{corr.notinsupp} that $a'\not\in\supp(x)$, a contradiction. 

Thus, $\theory T\mentfin \delta\act Y=Y$. 
\item
$\theory T\not\ment \delta\act Y=Y$.
To see this consider the elements $x_i=\{(\pi\fcomp\delta^i)\act \atomsdown\mid \pi\text{ finite}\}$ with the pointwise action, for every $i\in\mathbb Z$ (where $\mathbb Z$ is the integers; see Definition~\ref{defn.NZ}).

It is a fact that $(b\ a)\act x_i=x_i$, but it is also a fact that 
$\delta\act x_i=x_{i+1}\neq x_i$.
We interpret $\tau$ to be the set $\{x_i \mid i\in\mathbb Z\}$ and see that $(b\ a)\act x_i=x_i$ for every $i$ so the axiom $(b\ a)\act X=X$ is valid, but $\delta\act x_0\neq x_0$ so $\theory T\not\ment\delta\act X=X$.

This observation is exactly the \emph{fuzzy support} noted in \cite{gabbay:genmn}, see also Remark~3.3 from \cite{gabbay:pernl-jv}.
\end{itemize*}
\end{proof}

\begin{rmrk}
\label{rmrk.where.fail}
Where do the proofs fail?
Failure occurs in the interaction of $\atoms(r)$ (Definition~\ref{defn.atomsof}) with Theorem~\ref{thrm.pull.out.l}.

The reasonable definition for $\atoms(\shift\act X)$ is $\nontriv(\shift)/\pmss{X}$, which is infinite.
This gives us infinitely many atoms to avoid in order to guarantee $\atoms(r)\cap\supp(l)=\varnothing$ in Theorem~\ref{thrm.pull.out.l}.
Thus, $\supp([l]\denot{\interp H}{\varsigma}{r})$ need not be finite.
\end{rmrk}
 
As a corollary we can clarify something that is evident but only semi-formal in previous work: permissive-nominal terms with $\shift$ are strictly more expressive than permissive-nominal terms without $\shift$, and also strictly more expressive than `ordinary' nominal terms.

\section{Conclusions}
\label{sect.conclusions}

We have seen permissive-nominal sets and how, given a permissive-nominal set $\ns X$, we can build a corresponding nominal set $[\mathbb L]\ns X$ from $\ns X$ by applying to each $x\in|\ns X|$ an infinite simultaneous atoms-abstraction abstracting all but finitely many of the atoms in $\supp(x)$.
We have used this to translate between interpretations with differently sized sets of support, and we have used this translation to translate between different notions of validity for permissive-nominal syntaxes.

It can be easier to work with permissive-nominal models---even dramatically easier.
To see an example, compare the direct completeness proof for nominal algebra with respect to finitely-supported models in \cite{gabbay:nomuae} (subsections~4.3 and~4.4; over five pages including a `trick') with the completeness proof for the permissive-nominal algebra used in this paper with respect to permissive-nominal models in \cite{gabbay:nomtnl} (subsection~7.5; under two pages, and the maths is straightforward).
Even more extreme, compare the proof of the Nominal HSPA theorem from \cite{gabbay:nomahs} (twenty-eight pages) with the permissive-nominal HSP theorem from \cite{gabbay:nomtnl} (five pages).\footnote{This is unfair.  For instance, the five pages do not include setting up the syntax.  Still, looking at the maths, a leap in difference in complexity is clear.} 

What this means is that---based on this author's experience---even if the reader is interested specifically in finitely-supported models, it might be shorter and cleaner to prove completeness with respect to some flavour of infinitely-supported permissive-nominal models first, and then to use this paper off-the-shelf. 

\paragraph*{de Bruijn indexes}
The technical construction at the heart of this paper, $[\mathbb L]\ns X$ from Definition~\ref{defn.abstraction.X}, is arguably reminiscent of de Bruijn indexes \cite{bruijn:lamcnn}.
Given an $x\in|\ns X|$ with infinite support, we form $[l]x$ where $\supp(x)\setminus\supp(l)$ is finite.
In doing this, we in effect convert all but finitely many of the atoms in $\supp(x)$ into numerical indexes, where $a$ is identified by the position in $l$ at which it occurs (if any).

Of course this is not a literal replacement in $x$, since we assume no internal structure.
But since nominal elements have names, binding these names in order corresponds to turning them into numerical indexes.
An explicit connection is made in \cite[Section~4]{gabbay:genmn} as mentioned below.  

\paragraph*{Infinite sets of atoms in the literature}
The notions of infinite support, infinite lists of atoms, and infinite simultaneous atoms-abstraction were considered by the author in \cite{gabbay:fmhotn,gabbay:genmn}.

Translations between nominal abstract syntax, name-carrying syntax, and de Bruijn syntax were given in \cite[Section~4]{gabbay:genmn}.
The precise definition used in this paper is different and tailored to our intended application (we restrict to the subset of abstractions such that $\supp([l]x)$ is finite), and of course, we concentrate on things other than abstract syntax.

The notion of not-necessarily-finite support was also raised in \cite{pitts:nomlfo-jv}, and Cheney took up the suggestion in \cite{cheney:comhtn}.
The \emph{support ideals} considered there are not quite the same as the permission sets considered here (for instance, permission sets in this paper are never finite, and the emphasis on well-orderings is absent in Cheney's work), but the spirit of the maths is similar. 

\paragraph*{Implicit connection with unknowns}
A non-evident connection exists between $[l]x$ and unknowns, which this paper has not explored.
In a separate paper we demonstrate how a model of unknowns $X$ is given by infinite well-orderings of permission sets \cite{gabbay:twolns}.

One way to view $[l]x$ is as `$x$ abstracted by an unknown $X$'.

Now $\mathbb L$ from Definition~\ref{defn.L} is a \emph{single permutation orbit} under finite permutations of some list $l_\ast$ of atoms.
This author calls this a \emph{namespace}---$\mathbb L$ is a namespace, that is, a set of sets of atoms (in order) obtained by permuting \emph{finitely} many of them at a time.
We go from $\ns X$ to $[\mathbb L]\ns X$ essentially by abstracting a namespace, and because an unknown identifies a namespace, this can be read as a (level 1) abstraction of (the atoms in) a level 2 variable.

This paper does not make anything of these connections, but they exist in the background.
At some point, we hope to produce a broader account which will bring the threads in the various papers together and makes clearer the overall picture.
For now, the results in this paper have independent interest as discussed above. 

\paragraph*{\emph{shift}-permutations}
We concluded the technical part of this paper in Section~\ref{sect.shift} by discussing \emph{shift}-permutations.
These infinite permutations are useful and mathematically well-behaved, but they mark a point at which permissive-nominal techniques go strictly beyond the expressivity of nominal techniques, and this is made formal: we saw in Section~\ref{sect.shift} how the results of Section~\ref{sect.completeness.fin} depend on permutations $\pi$ being finite and when we include infinite permutations in syntax, the results begin to fail.
This is reasonable and as it should be.

So a lesson we can draw from this paper and from the translation in \cite[Section~4]{gabbay:perntu-jv} is this: permissive-nominal terms with finite permutations are essentially the same thing as (but somewhat better-behaved than) `ordinary' nominal terms; permissive-nominal terms with possibly infinite permutations are different, and they are strictly more expressive.

\paragraph*{Non-equivariance of $\atoms$}
One curious aspect of our proofs is that the function $\atoms$ from Definition~\ref{defn.atomsof}, which plays a role in Section~\ref{sect.two.comm}, is not invariant under $\alpha$-equivalence.
For instance, $\atoms([a]X)=\{a\}$ and $\atoms([b](b\ a)\act X)=\{b,a\}$ (where $a\in\atomsdown$ and $b\in\atomsup$).

It is shown in \cite{gabbay:twolns} and \cite{gabbay:nomtnl} that valuations $\varsigma$ (Definition~\ref{defn.valuation}) can be thought of as (finite-)equivariant maps out of \emph{moderated} unknowns $\pi\act X$ considered as a permissive-nominal set.

The non-equivariance of $\atoms$ is an artefact of the fact that the syntax of this paper amounts to choosing $\id\act X$ as a representative of the permutation equivalence class $\{\pi\act X\mid \pi\text{ a permutation}\}$.
Permissive-nominal syntax is \emph{already} non-equivariant, because for each unknown-up-to-permutations we have chosen a canonical representative $X$.

None of this matters for the proofs here, because we only care about avoiding capture with finitely many atoms.

\paragraph*{Final words on set theory}
The results of this paper are reminiscent of the upwards and downwards L\"owenheim-Skolem theorems, which express that a first-order theory cannot `count' the cardinality of its infinite models \cite{hodges:modt}.
The construction of this paper can be read as saying that first-order permissive-nominal syntax with finite permutations cannot `count' the cardinality of its supporting sets. 

We believe it would be fairly easy to strengthen and generalise this result to the following: first-order nominal syntax cannot `count' the cardinality of the set of atoms or its supporting sets, so long as these are no smaller than the permutations in that syntax.
Making this formal would require us to be just a little systematic 
but it should not be too hard using a free construction---and the syntax should be a natural generalisation sufficient to subsume permissive-nominal algebra and permissive-nominal logic.

More generally, we can also ask how the group of permutations can be independently enlarged or restricted in syntax and in the denotation.
For instance, in this paper we have considered syntax and semantics using the \emph{same} group of permutations.
But the semantics could allow more permutations than the syntax, e.g. we could allow \emph{shift} in the denotation (this is useful to `make support smaller' in some element), but not in the syntax (so that we might avoid the issues discussed in Subsection~\ref{subsect.shift} and specifically in Remark~\ref{rmrk.where.fail}).
In short, we see this paper as the first of a family of similar results which may become useful if and when further variations on logics based on nominal terms, and their models, are developed. 
We leave these thoughts to future work. 

\subsection*{Acknowledgements.}  
Many thanks to an anonymous referee, without whose careful input this paper would not have reached its current form.

\newpage 

\end{document}